%% file: 20211014_SS_SPIR_JSAC2.tex
\documentclass[journal,comsoc]{IEEEtran}

\usepackage[T1]{fontenc}

\usepackage{graphicx}
\usepackage{float}

\usepackage{amssymb,amsmath}
\usepackage{amsthm}
\usepackage{latexsym, bm, mathtools, braket, multirow,enumitem}
\usepackage[cmintegrals]{newtxmath}
\input{symbols_JSAC.tex}

\usepackage{array}
\newcolumntype{M}[1]{>{\centering\arraybackslash}m{#1}}
\newcolumntype{N}{@{}m{0pt}@{}}

\restylefloat{table,figure}

\begin{document}

\title{Equivalence of Non-Perfect Secret Sharing and Symmetric Private Information Retrieval with General Access Structure}

\author{Seunghoan~Song,~\IEEEmembership{Member,~IEEE}
        and
        Masahito~Hayashi,~\IEEEmembership{Fellow,~IEEE}
\thanks{S. Song is supported by JSPS Grant-in-Aid for JSPS Fellows No. JP20J11484.  
M. Hayashi is supported in part by Guangdong Provincial Key Laboratory (Grant No. 2019B121203002),
a JSPS Grant-in-Aids for Scientific Research (A) No.17H01280 and for Scientific Research (B) No.16KT0017, and Kayamori Foundation of Information Science Advancement.
This article was presented in part at Proceedings of 2021 IEEE International Symposium on Information Theory \cite{SH21}.}
\thanks{S. Song is with Graduate school of Mathematics, Nagoya University, Nagoya, 464-8602, Japan
(e-mail: m17021a@math.nagoya-u.ac.jp).}
\thanks{M. Hayashi is with 
Shenzhen Institute for Quantum Science and Engineering, Southern University of Science and Technology,
Shenzhen, 518055, China,
Guangdong Provincial Key Laboratory of Quantum Science and Engineering,
Southern University of Science and Technology, Shenzhen 518055, China,
Shenzhen Key Laboratory of Quantum Science and Engineering, Southern
University of Science and Technology, Shenzhen 518055, China,
and Graduate School of Mathematics, Nagoya University, Nagoya, 464-8602, Japan
(e-mail:hayashi@sustech.edu.cn).}
}

\maketitle

\begin{abstract}
We study the equivalence between non-perfect secret sharing (NSS) and symmetric private information retrieval (SPIR) with arbitrary response and collusion patterns. 
NSS and SPIR are defined with an access structure, which corresponds to the authorized/forbidden sets for NSS and the response/collusion patterns for SPIR.
We prove the equivalence between NSS and SPIR in the following two senses. 
1) Given any SPIR protocol with an access structure, an NSS protocol is constructed with the same access structure and the same rate. 
2) Given any linear NSS protocol with an access structure, a linear SPIR protocol is constructed with the same access structure and the same rate.
We prove the first relation even if the SPIR protocol has imperfect correctness and secrecy.
From the first relation, 
    we derive an upper bound of the SPIR capacity for arbitrary response and collusion patterns. 
For the special case of $\mathsf{n}$-server SPIR with $\mathsf{r}$ responsive and $\mathsf{t}$ colluding servers,
    this upper bound proves that the SPIR capacity is $(\mathsf{r}-\mathsf{t})/\mathsf{n}$.
From the second relation, we prove that a SPIR protocol exists for any response and collusion patterns.
\end{abstract}


%
\IEEEpeerreviewmaketitle

\section{Introduction}

\subsection{Nonperfect Secret Sharing and Symmetric Private Information Retrieval}

	Secret sharing (SS) and private information retrieval (PIR) are two extensively studied cryptographic protocols.
	SS \cite{Shamir79,Blakley79} considers the problem in which a dealer encodes a secret into $\snn$ shares so that some subsets of shares can reconstruct the secret but the other subsets have no information of the secret.
	PIR \cite{CGKS98} considers the problem in which a user retrieves one of the multiple files from server(s) without revealing which file is retrieved.
	Since PIR with one server has no efficient solution \cite{CGKS98}, it has been extensively studied with multiple non-communicating servers, and thus, in the following, we simply denote multi-server PIR by PIR.
	SS and PIR have a similar structure because 
		the secrecy of both protocols is obtained by 
			partitioning the confidential information.
	On the other hand, the two protocols have a different structure because
        in SS, the secret is both the confidential and targeted information 
        but in PIR, the targeted file is not confidential.
	From the similarity, there have been several studies to construct PIR from secret sharing \cite{GIKM00, BIKO12,LMD17,YSL18, DER18}.
	However, the relation between these two protocols has not been clearly discussed.
	
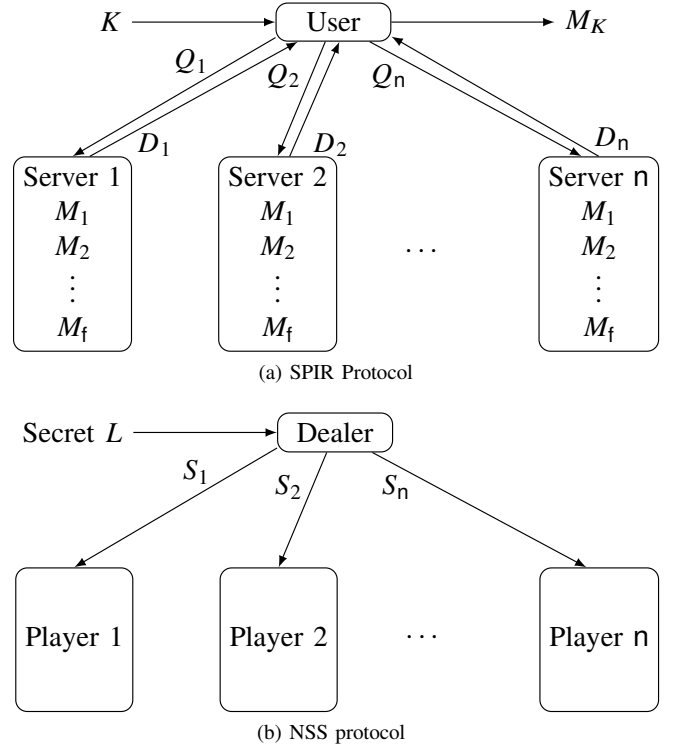
\begin{figure}
\subfloat[SPIR Protocol]{
\centering
\resizebox {1\linewidth} {!} {
\begin{tikzpicture}[node distance = 3.3cm, every text node part/.style={align=center}, auto]
    \node [block] (user) {User};
    \node [left=5em of user] (sent) {${K}$};
    \node [block,minimum height = 5em, below left=4em and 5em of user] (serv1) {$\text{Server 1}$\\{$M_1$}\\{$M_2$}\\\vdots\\{$M_{\sff}$}};
    \node [block,minimum height = 5em, right=3em of serv1] (serv2) {$\text{Server 2}$\\{$M_1$}\\{$M_2$}\\\vdots\\{$M_{\sff}$}};
    \node [right=2em of serv2] (ten) {\textbf{$\cdots$}};
    \node [block,minimum height = 5em, right=3em of ten] (servn) {$\text{Server }{\snn}$\\{$M_1$}\\{$M_2$}\\\vdots\\{$M_{\sff}$}}; 

    \node [right=2cm of user] (receiv) {${M_{K}}$}; 
    
    
    \path [line] (sent) -- (user);
    
    \path [line] (user.195) --node[pos=0.2,left=2mm] {${Q_1}$} (serv1.north);
    \path [line] (user) --node[pos=0.3,left] {${Q_2}$} (serv2.north);
    \path [line] (user) --node[pos=0.3,left=2mm] {${Q_{\snn}}$} (servn.100);
    
    \path [line] (serv1.80)--node[pos=0.1,right=2mm] {{$D_1$}} (user);
    \path [line] (serv2.83) --node[pos=0.1,right=1mm] {{$D_2$}} (user.287);
    \path [line] (servn.north) --node[pos=0.15,right=2mm] {{$D_{\snn}$}} (user.346);
    
    \path [line] (user.east) --  (receiv);
\end{tikzpicture}
}
} \\
\subfloat[NSS protocol]{
\centering
\resizebox{1\linewidth}{!}{
\begin{tikzpicture}[node distance = 3.3cm, every text node part/.style={align=center}, auto]
    \node [block] (user) {Dealer};
    \node [left=5em of user] (sent) {Secret $L$};
    \node [block,minimum height = 5em, below left=4em and 5em of user] (serv1) {$\text{Player }1$};
    \node [block,minimum height = 5em, right=3em of serv1] (serv2) {$\text{Player }2$};
    \node [right=2em of serv2] (ten) {\textbf{$\cdots$}};
    \node [block,minimum height = 5em, right=3em of ten] (servn) {$\text{Player }{\snn}$};
    
    \path [line] (sent) -- (user);
    
    \path [line] (user.195) --node[pos=0.2,left=2mm] {${S_1}$} (serv1.north);
    \path [line] (user) --node[pos=0.3,left] {${S_2}$} (serv2.north);
    \path [line] (user) --node[pos=0.3,left=2mm] {${S_{\snn}}$} (servn.100);
    
\end{tikzpicture}
}}
\caption{NSS protocol and SPIR protocol with access structure $(\fA,\fB)$.
    Let $\snn = 3$, $\fA = \{  \{2,3\}, \{1,2,3\} \}$, and $\fB = \{\emptyset,  \{1\},  \{2\}, \{3\}, \{1,2\} \}$.
    Since $\{2,3\}\in\fA$, $M_K$ ($L$) can be reconstructed from the answers $D_2,D_3$ (shares $S_2,S_3$).
    Since $\{1,2\}\in\fB$, no information of $K$ ($L$) can be extracted from the answers $D_1,D_2$ (shares $S_1,S_2$).
    } \label{fig:compintro}
\end{figure}

\begin{table}[t] 
\begin{center}
\caption{Comparison of SPIR and NSS} \label{tab:compintro}

\begin{tabular}{|c|c|c|}
                        \hline
                        & SPIR & NSS    \\
                        \hline
Messages & $\sff$ files &  one secret  \\
                        \hline
Participants & one user \& $\snn$ servers & one dealer \& $\snn$ players  \\
                        \hline
\multirow{3}{*}{Security constraints} & Correctness &  Correctness \\
& Server secrecy &  Secrecy \\
& User secrecy  & \\
\hline
Communication & Two-way & One-way  \\
    \hline
\end{tabular}

\end{center}
\end{table}

In this paper, we study the relation between an extended class of SS, called {\em non-perfect SS (NSS)}, and that of PIR, called {\em symmetric PIR (SPIR)}.
NSS is first discussed by \cite{BM85,Yamamoto86} with thresholds
    and is extended for general access structures \cite{OKT93,OK95,Paillier98,Yoshida19, FHKP17}. 
NSS with general access structures is defined with two collections $\ACC, \REJ \subset 2^{[\snn]}\coloneqq 2^{\{1,\ldots,\snn\}}$
	which represent the {\em authorized sets} and the {\em forbidden sets}, respectively.
The shares indexed by any element of $\ACC$ can reconstruct the secret
	but those indexed by any element of $\REJ$ have no information of the secret.
	We denote NSS with $\ACC$ and $\REJ$ by $(\ACC,\REJ,\snn)$-NSS.
The (perfect) SS corresponds to the $(\ACC,\REJ,\snn)$-NSS with $\ACC \cup \REJ = 2^{[\snn]}$.

SPIR \cite{GIKM00} is a variant of PIR in which the user only obtains the targeted file but no information of the other files.
The paper \cite{GIKM00} proved that shared randomness of servers is necessary to achieve SPIR.
This paper considers SPIR with arbitrary response and collusion patterns $\REC,\COL\subset 2^{[\snn]}$. 
We define a $(\REC,\COL,\snn,\sff)$-SPIR protocol by an $\snn$-server $\sff$-file protocol with the following conditions:
		i) the user correctly recovers the targeted file
			even if only the servers indexed by $\cA\in \REC$ answer to the user;
		ii) the identity of the retrieved file is not leaked even if the servers indexed by $\cB\in\COL$ collude;
		and
		iii) the user obtains no other information than the targeted file.
If $\REC=  \{ [\snn] \}$ and $\COL = \{ \{i\} \in i \in [\snn] \}$,
	$(\REC,\COL,\snn,\sff)$-SPIR is the SPIR in which no servers collude and all servers respond, 
		which has been discussed in \cite{GIKM00, SJ17-2}.

The efficiency of a PIR protocol is evaluated by the {\em PIR rate}
	\begin{align}
	\RPIR = \frac{\text{(Size of the targeted file)}}{\text{(Total size of the answers)}} ,
	\label{defeq:PIR}
	\end{align}
and the optimal efficiency is derived as the supremum of the PIR rate, which is called the {\em PIR capacity} \cite{SJ17}, and has been studied in many papers.
The paper \cite{JSJ17} derived the capacity of PIR with disjoint collusion patterns, which is a special case of the arbitrary collusion patterns. 
This model is extended to the arbitrary collusion patterns by \cite{TGKFHE17,ZG17,SJ18-col}, where PIR protocols are constructed on coded data storage. %
On the replicated data storage, the capacity of PIR with arbitrary collusion patterns \cite{YLK20}
    and the capacity of SPIR with arbitrary collusion and eavesdropping patterns \cite{CLK20}
    are characterized by the solution of a linear programming problem. 
For the PIR/SPIR with thresholds, i.e., $\fA$ ($\fB$) consists of all subsets with cardinality at least $\srr$ (at most $\stt$),
    the papers \cite{BU19-2}, \cite{WS18}, and \cite{HFLH21} derived
    the capacity of PIR with $\stt$ colluding and $\sbb$ byzantine servers,
    the capacity of SPIR with $\stt$ colluding and $\see$ eavesdropping servers with an assumption on the shared randomness of the servers,
    and 
    the PIR/SPIR capacity with $\stt$ colluding, $\see$ eavesdropping, $\srr$ responsive, and $\sbb$ byzantine servers, under various assumptions on the protocols and the shared randomness, respectively.
The paper \cite{TGKFH19} constructed PIR/SPIR protocols with $\stt$ colluding, $\srr$ responsive, and $\sbb$ byzantine servers. 
However, no existing study discussed SPIR with arbitrary response and collusion patterns.

\subsection{Main Results}

As the first result, this paper shows that $(\REC,\COL,\snn,\sff)$-SPIR implies $(\ACC,\REJ,\snn)$-NSS.
	To state this result, we formally define $(\REC,\COL,\snn,\sff)$-SPIR and $(\ACC,\REJ,\snn)$-NSS protocols with incomplete security by the measures of correctness and secrecy.
	With abuse of notation, we prove that given an $(\REC,\COL,\snn,\sff)$-SPIR protocol with nearly complete security, we propose a method to construct an $(\ACC,\REJ,\snn)$-NSS protocol with nearly complete security.
	The proof idea is simply described as follows. 
	To construct an $(\ACC,\REJ,\snn)$-NSS protocol from a given $(\REC,\COL,\snn,\sff)$-SPIR protocol, the dealer simulates the $(\REC,\COL,\snn,\sff)$-SPIR protocol 
		while setting the secret of NSS as one of the files,
		and encodes the $\snn$ answers as $\snn$ shares. 
	Then, the shares indexed by $\ACC$ reconstruct the secret by the SPIR's correctness.
    For the secrecy part, we prove that the shares indexed by $\COL$ have no information of the secret from the SPIR's two secrecy conditions.


    When the SPIR and NSS protocols have complete correctness and secrecy,
        these protocols are called {\em completely secure} and denoted by CSSPIR and CSNSS protocols, respectively.
	One interesting corollary of our first result is 
		an upper bound on the CSSPIR capacity with arbitrary response and collusion patterns.
	Similar to the PIR rate defined in \eqref{defeq:PIR}, the SS rate\footnote[1]{The efficiency of SS is often considered with {\em information rate} and {\em information ratio} \cite{Beimel11}.
	The information rate is defined by replacing the denominator of \eqref{RSSdef} by the maximum size of a share, and the information ratio is the inverse of the information rate.
	However, we define the SS rate for the correspondence with the PIR rate.}
	 $\RSS$ is defined as 
	\begin{align}
	\RSS = \frac{\text{(Size of the secret)}}{\text{(Total size of the shares)}}.
	\label{RSSdef}
	\end{align}
	In our conversion from SPIR to NSS, 
		any SPIR protocol with PIR rate $\RPIR$ is converted into an NSS protocol with SS rate $\RSS = \RPIR$.
	Furthermore, any $(\ACC,\REJ,\snn)$-CSNSS protocols satisfy $\snn\RSS \leq \delta(\ACC,\REJ) \coloneqq \min \{ |\cA-\cB| \mid \cA\in \ACC, \ \cB\in \REJ \}$ \cite{OKT93,OK95,Paillier98}.
	Thus, we obtain $\snn\RPIR \leq \delta(\ACC,\REJ)$ for any $(\REC,\COL,\snn,\sff)$-CSSPIR protocols.
    This is the first result to characterize the CSSPIR capacity of arbitrary collusion and response patterns.

   %

	As a special case, when $\REC$ ($\COL$) consists of all subsets of $[\snn]$ with cardinality at least $\srr$ (at most $\stt$),
		we obtain $\delta(\ACC,\REJ) = \srr-\stt$, i.e., 
		the rate of $(\REC,\COL,\snn,\sff)$-CSSPIR is upper bounded by $(\srr-\stt)/\snn$.
	This special case generalizes the result of Holzbaur et al. \cite{HFLH21}, which proved the same upper bound for a restricted class of CSSPIR, i.e., linear CSSPIR with additive randomness.
	Since the protocol by Tajeddine et al. \cite{TGKFH19} achieves this upper bound\footnote[2]{In \cite{TGKFH19}, 
	the notation $\srr$ is for the number of {\em unresponsive} servers, which is $\snn-\srr$ in our paper.
	Moreover, \cite{TGKFH19} defined the denominator of the PIR rate \eqref{defeq:PIR} as ``bit size of all responsive servers''.}, 
		our upper bound proves that the capacity of CSSPIR with $\srr$ responses and $\stt$ colluding servers is $(\srr-\stt)/\snn$.

Our second main result is the equivalence of linear $(\REC,\COL,\snn,\sff)$-CSSPIR and linear $(\ACC,\REJ,\snn)$-CSNSS.
    Linear $(\REC,\COL,\snn,\sff)$-SPIR (linear $(\ACC,\REJ,\snn)$-NSS) is a well-known class of SPIR (NSS) in which the answers (shares) are generated by linear encoders.
	To prove this result, we define {\em multi-target monotone span programs (MMSP)} with general access structures $(\ACC,\REJ)$, which we call {\em $(\ACC,\REJ,\snn)$-MMSP}.
	An MMSP is a pair of a matrix and a map with two linear-algebraic conditions.
    To prove the equivalence result, 
        we separately prove that 
            a $(\ACC,\REJ,\snn)$-CSNSS protocol, a $(\ACC,\REJ,\snn)$-MMSP, and a $(\ACC,\REJ,\snn)$-CSSPIR protocol,  
            are constructed 
            from 
            a $(\ACC,\REJ,\snn)$-CSSPIR protocol, a $(\ACC,\REJ,\snn)$-CSNSS protocol, and a $(\ACC,\REJ,\snn)$-MMSP, respectively,
        which implies the equivalence of the three protocols.
    Our equivalence result of $(\ACC,\REJ,\snn)$-CSNSS and $(\ACC,\REJ,\snn)$-MMSP
        generalizes the equivalence of completely secure SS and monotone span program \cite{Brickell89, KW93,Beimel_thesis}.
Since there exists a linear CSNSS protocol for any access structure $(\fA,\fB)$ \cite{FHKP17},
    the same existence holds for CSSPIR by our second result.

\subsection{Basic Idea in Equivalence} \label{susec:basicidea}
    
The key intuition for the equivalence is 
    the similarity in the structure of the protocols (see Figure~\ref{fig:compintro} and Table~\ref{tab:compintro}).
From this similarity, the server secrecy of SPIR
    is closely related to the secrecy of NSS, 
    while the user secrecy of SPIR is not covered in NSS.
Thus, we can roughly state that SPIR is more complicated protocol than NSS protocol.
With this idea, conversion from the complicated SPIR protocol to the simpler NSS protocol is obtained by utilizing only a part of the SPIR protocol as the NSS protocol.
That is, we only need the operational properties of SPIR, rather than some algebraic structure, for converting it into NSS.
As a result, the conversion from SPIR to NSS can be completed without any linearity assumption.

On the other hand, to construct conversion from NSS to SPIR, we need to add the user secrecy to the NSS protocol.
This addition of user secrecy is hard to be implemented only with the operational definition of NSS protocols.
Thus, instead, we convert linear NSS into linear SPIR.
With the linearity algebraic structure in linear NSS, we add the user secrecy to NSS so that it operates the SPIR's tasks.
As a result, the conversion from NSS to SPIR is limited to linear protocols in this paper.

One interesting observation of the conversion from NSS to SPIR is that 
    the one message protocol, NSS,
    is evolved into the multiple message protocol, SPIR.
This task is completed in our conversion from linear NSS to linear SPIR by the following idea.
When $k$-th file $M_k$ is the desired file of the user, 
    the user forces the servers to choose $M_k$ as a secret of the NSS without leaking $k$ and answer the generated shares of the NSS.
This forcing step by the user is accomplished with a well-established query structure.
Finally, the user reconstructs $M_k$ by collecting an authorized set of shares (answers).

\subsection{Organization}


The remainder of the paper is organized as follows.
Section~\ref{sec:defi} gives the definitions of NSS, SPIR, and related concepts.  
Section~\ref{sec:main_result} presents the main results of the paper.
Section~\ref{sec:SPIRtoNSS} proves the conversions from SPIR to NSS and from linear SPIR to linear NSS.
Sections~\ref{sec:NSS=MMSP} and \ref{sec:NSStoSPIR} prove the conversions from linear CSNSS to MMSP and from MMSP to linear CSSPIR, respectively.
Section~\ref{sec:examples} 
    give two examples of the constructions: 
        a construction with general non-threshold access structure,
        and 
        optimal constructions of linear CSNSS and linear CSSPIR when the access structure with thresholds.
Section~\ref{sec:conclusion} is the conclusion of the paper.

\subsubsection*{Notations}
	We denote random variables by uppercase letters (e.g., $A,B$), 
		and values of the random variables by lowercase letters (e.g., $a,b$),
		sets by calligraphy letters (e.g., $\cA,\cB$),
		parameters in protocols by sans serif lowercase letters (e.g., $\saa, \sbb$),
		and matrices by sans serif upper letters (e.g., $\sA,\sB$).
	We also denote $[n,m] = \{n,n+1,\ldots, m\}$ and $[n] = [1:n]$.
	For any set or sequence $A = \{A_1,\ldots, A_n\}$ and any $\cX \subset [n]$, 
		we denote $A_{\cX} \coloneqq \{ A_i \mid i \in \cX \}$.
	For any $ n\times m$ matrix $\sA$ and any $\cX \subset [n]$, 
		we denote $\sA_{\cX}$ is the restricted matrix by the rows indexed by $\cX$.
	The finite field of order $q$ is denoted by $\FF_q$.
	For a set $\cA$, $\id_{\cA}$ denotes the identity map on $\cA$.
	For a random variable $X$, $\pr_X [ f(X) ]$ is the probability that $X$ satisfies the condition $f(X)$.


%
%



\begin{table}[t]
\begin{center}
\caption{Definition of symbols} \label{tab:symbols}
\begin{tabular}{|c|c|c|c|c|}
\hline
Symbol &	SPIR	&	Secret Sharing \\
\hline
\hline
$\snn$	&	Number of servers	&	Number of shares \\
\hline
$\sff$	&	Number of files & 	- 	\\
\hline
$\smm$	&	Size of one file 		& 	Size of secret	\\
\hline
$\srr$	&	Number of responsive servers	&	Reconstruction threshold	\\
\hline
$\stt$	&	Number of colluding servers		&	Secrecy threshold	\\
\hline
\end{tabular}
\end{center}
\end{table}

\section{Formal definition of NSS and SPIR} \label{sec:defi}

In this section, we formally define SPIR, NSS, and MMSP.
For these definitions, we first define access structures.
\begin{defi}[Access structure]
Let $\snn$ be a positive integer.
We call $\fA \subset 2^{[\snn]}$ an {\em monotone increasing collection} 
	if $\cA \in \fA$ implies $\cC\in\fA$ for any $\cA \subset \cC \subset[\snn]$.
In contrast, we call $\fB \subset 2^{[\snn]}$ a {\em monotone decreasing collection} 
	if $\cB \in \fB$ implies $\cC\in\fB$ for any $\cC \subset \cB$.
An {\em access structure} on $[\snn]$ is defined as a pair of monotone increasing and decreasing collections $\REC$ and $\COL \subset 2^{[\snn]}$ such that $\REC \cap \COL = \emptyset$.
\end{defi}

\begin{exam} \label{exam:as}
When $\snn = 3$, $\fA = \{  \{2,3\}, \{1,2,3\} \}$, and $\fB = \{\emptyset,  \{1\},  \{2\}, \{3\}, \{1,2\} \}$,
    the pair $(\fA,\fB)$ forms an access structure.
The monotone increasing collection $\fA$ consists of all subsets containing $\{2,3\}$,
    and
    the monotone decreasing collection $\fB$ consists of all subsets contained in $\{1,2\}$ or $\{3\}$.
In this example, the subsets of size $2$ are contained in different collections:
    $\{2,3\}\in\fA$, $\{1,2\}\in\fB$, and $\{1,3\} \not \in \fA\cup \fB$.
\end{exam}

\subsection{Formal definition of symmetric private information retrieval (SPIR)}
    \label{sec:spirforam}


We formally define a SPIR protocol with one user and $\snn$ servers as follows. 

\begin{defi}[$(\REC,\COL,\snn,\sff)$-SPIR] \label{defi:spir}
Files are given as a uniform random variable $M = (M_1,\ldots,M_{\sff}) \in \cM^{\sff}$
	and $\smm \coloneqq |\cM|$.
Each of $\snn$ servers contains the files $M$.
Let $\RandSPIR \in \SRandSPIR$ be uniform random variable, called the {\em random seed} for servers, and the random seed $\RandSPIR$ is encoded as $h(\RandSPIR) = T = (T_1,\ldots, T_n) \in \cT \subset \cT_1\times \cdots \times \cT_{\snn}$ by the shared randomness encoder $h$.
The randomness $T$ is distributed so that $j$-th server contains $T_j$.
Let $K$ be a uniform random variable in $[\sff]$, which represents the user's input.
When $K=k$, the targeted file is $M_k$.

A protocol $\PSPIR$ is defined by the deterministic mappings $f,g_1,\ldots, g_{\snn}$ in the following steps.
\begin{itemize}[leftmargin=1.5em]
\item \textbf{Query}:
Depending on the user's private uniform randomness $R\in\cR$, 
	the user prepares $\snn$ queries by
\begin{align}
Q = (Q_1,\ldots, Q_{\snn}) = f(K,R) \in \cQ_1\times \cdots \times \cQ_{\snn}
\end{align}
and sends the $j$-th query $Q_j$ to the $j$-th server.

\item \textbf{Answer}:
The $j$-th server returns 
\begin{align}
D_j = g_j(Q_j,M, T_j) \in \cD_j
\end{align}
to the user.

	
\end{itemize}

For an access structure $(\REC, \COL)$ on $[\snn]$.
    the protocol $\PSPIR$ is called an $(\REC,\COL,\snn,\sff)$-SPIR protocol with security $(\alpha,\beta,\gamma)$ if 
    the following conditions are satisfied.
\begin{itemize}[leftmargin=1.5em]
\item \textbf{Correctness}:
    We define the user's maximum likelihood decoder 
    \begin{align*}
    \hat{m}^{\mathrm{ML}}_{d_{\cA},r,k} &\coloneqq \argmax_{m_k\in\cM} \Pr[D_{\cA}= d_{\cA}|(M_K, R,  K) = (m_k, r, k) ].
    \end{align*}
Then, 
\begin{align}
\alpha \ge 
\alpha (\PSPIR) \coloneqq \max_{r\in\cR,k\in[\sff],\cA\in\fA} \pr_{M_k,D_{\cA}} [ M_k \neq \hat{m}^{\mathrm{ML}}_{D_{\cA},r,k}].
\label{correct:spir}
\end{align}

\item \textbf{Server secrecy}:
\begin{align}
\beta \ge 
\beta (\PSPIR) \coloneqq
\max_{r\in\cR, k\in[\sff] } I(D ; M_{[\sff]\setminus \{k\}} | R=r, K=k) .
\label{cond:serv_secrecydef}
\end{align}

\item \textbf{User secrecy}:
\begin{align}
\gamma \ge
\gamma(\PSPIR)  \coloneqq \max_{\cB\in \COL} I(K;Q_{\cB}).
\end{align}

\end{itemize}

The PIR rate of the protocol $\PSPIR$ is defined as 
	\begin{align}
	\RPIR(\PSPIR) \coloneqq  
		\frac{\log \smm}{ \sum_{j=1}^{\snn}\log |\cD_j|}.
	\end{align}
The shared randomness rate is defined as 
\begin{align}
R_{\RandSPIR} = \frac{\log |\SRandSPIR|}{\log \smm}.
\end{align}
\end{defi}

%
%

\begin{remark}
From the definition of the protocol, 
    it seems natural to define the server secrecy condition as $\beta' (\PSPIR) \coloneqq\max_{\cA\in\fA, r\in\cR, k\in[\sff] } I(D_\cA ; M_{[\sff]\setminus \{k\}} | R=r, K=k)$, which is the maximization of the $\beta' (\PSPIR)$ defined in \eqref{cond:serv_secrecydef}. 
Indeed, we have $\beta' (\PSPIR) = \beta (\PSPIR)$ since the collection $\fA$ contains the set $[\snn]$ by the monotone increasing property.
\end{remark}


In the above security conditions,
$\alpha (\PSPIR)$ is the worst-case error probability with the maximum likelihood decoder,
$\beta (\PSPIR)$ is the worst-case leakage of the non-targeted files to the user,
and $\gamma (\PSPIR)$ is the worst-case leakage of the index $K$ to the colluding servers.
If $(\alpha,\beta,\gamma) = (0,0,0)$, 
	the $(\REC,\COL,\snn,\sff)$-SPIR protocol has complete security.
\begin{defi}[CSSPIR]
A  $(\REC, \COL,\snn,\sff)$-SPIR protocol with security $(\alpha,\beta,\gamma) = (0,0,0)$ is called 
    a completely secure $(\REC, \COL,\snn,\sff)$-SPIR ($(\REC, \COL,\snn,\sff)$-CSSPIR) protocol.
\end{defi}
In a $(\REC,\COL,\snn,\sff)$-CSSPIR protocol,
		the user can recover the targeted file $M_K$ without error from the answers indexed by any $\cA\in\REC$,
	but the servers indexed by $\cB\in\COL$ obtain no information of $K$.
The achievable rate and capacity of $(\REC, \COL,\snn,\sff)$-CSSPIR are defined as follows.

\begin{defi}[Achievability of $(\REC, \COL,\snn,\sff)$-CSSPIR]
A PIR rate $R$ is achievable if
	there is a sequence of $(\REC,\COL,\snn,\sff)$-CSSPIR protocols
		$\{ \PSPIR \}_{\smm\in\mathbb{N}}$
	such that $\RPIR(\PSPIR) \to R$ as $\smm\to\infty$.
\end{defi}

\begin{defi}[Capacity of $(\REC,\COL,\snn,\sff)$-CSSPIR]
Capacity of $(\REC,\COL,\snn,\sff)$-CSSPIR, 
	$C^{(\REC,\COL,\snn,\sff)}_{\mathrm{SPIR}}$,
	is the supremum of 
	achievable rates of $(\REC,\COL,\snn,\sff)$-CSSPIR.
%
\end{defi}

As special cases, SPIR with thresholds and linear SPIR are defined as follows.

\begin{defi}[$(\srr,\stt,\snn,\sff)$-SPIR]
When $\REC = \{ \cA \subset [\snn] \mid |\cA| \geq \srr \}$
	and $\COL = \{ \cB \subset [\snn] \mid |\cB| \leq \stt \}$,
	$(\REC,\COL,\snn,\sff)$-SPIR protocols are called $(\srr,\stt,\snn,\sff)$-SPIR protocols.
\end{defi}

\begin{defi}[{Linear SPIR}] \label{defi:lin_spir}
A protocol $\PSPIR$ is called a {\em linear SPIR} protocol if 
    the following conditions are satisfied with a map $\tau: [\sz] \to [\snn]$, called the {\em position map}.
\begin{description}

    \item[Vector representation of files]
    The files $M_i$ are written as a vector in $\Fq^{\sx}$. 
    The entire file is written by the concatenated vector $M = (M_1, \ldots, M_\sff) \in \Fq^{\sff \sx }$.

    \item[Linearity of shared randomness]
    The random seed $\RandSPIR$ is written by a uniform random vector in $\Fq^{\sy}$.
    The randomness encoder is written as a matrix $\sH\in\Fq^{\sz\times \sy}$ and 
    the shared randomness is written as $T = \sH \RandSPIR \in \Fq^{\sz}$.
    The randomness of the $j$-th server is written as $T_j = \sH_{\tau^{-1}(j)} T'  \in \Fq^{|\tau^{-1}(j)|}$.
    \item[Linearity of servers]
    The answer of the $j$-th server $D_j$ is written as 
    the sum of 
        the shared randomness $T_j$
        and
        the encoded output of the files $M$ by the linear function $ Q_{\tau^{-1}(j)} $, which depends on the query, i.e., 
        \begin{align}
        D_j = Q_{\tau^{-1}(j)} M  + T_j \in \Fq^{|\tau^{-1}(j)|}.
        \label{eq:linearencc}
        \end{align}
Therefore, we can consider that the query to the $j$-th server is given as
the linear function, a matrix, $ Q_{\tau^{-1}(j)} \in \Fq^{\sz\times \sff\sx }$.
\end{description}

A linear SPIR is described by the triplet $(\sH, Q,\tau)$ of randomness encoder $\sH\in\Fq^{\sz\times \sy}$, random matrix of query $Q\in \Fq^{\sz\times \sff\sx }$, and position map $\tau: [\sz] \to [\snn]$.
Due to the above conditions, the PIR rate and the shared randomness rate of a linear SPIR protocol are $\sx/\sz$ and $\sy/\sx$, respectively.
\end{defi}

\begin{remark}
Our definition of linear SPIR generalizes the definition of Holzbaur et al. \cite{HFLH21}. 
The linear SPIR defined in \cite{HFLH21} corresponds to our definition with 
    $|\tau^{-1}(j)| =c $ for a fixed $c$.
\end{remark}

In the definition of linear PIR,
    linearity is required for the generation of the shared randomness and the servers' answers.
We define a new subclass of linear PIR with a specific linear condition on the user's encoding as follows.

%

\begin{defi}[Projected linear SPIR] \label{defi:prj_lin_spir}
A linear SPIR protocol $\PSPIR$ is called a {\em projected linear SPIR} protocol if 
the query $
Q=(Q_{\tau^{-1}(1)}, \ldots, Q_{\tau^{-1}(\snn)} )$ satisfies the following condition;
\begin{description}
\item[Linearity of user]
The uniform randomness $R$ of the user 
is written as a matrix in $\FF_q^{\sy \times \sff\sx}$.
The query is written as 
	\begin{align}
	Q 
	= \sJ \sE_k + \sH R 
	\in\FF_q^{\sz\times \sff\sx},
	\end{align}
	where 
	$\sJ$ is a matrix in $\FF_q^{\sz\times \sx}$,
	$\sE_k\in\FF_q^{\sx\times \sff\sx} = \FF_q^{\sx\times \sx}\times \cdots \times \FF_q^{\sx\times \sx}$ is the matrix whose $k$-th block is the $\sx\times \sx$ identity matrix and the other blocks are zero,
	and
	$\sH\in\FF_q^{\sz\times \sy}$ is the same matrix as the shared randomness encoder.
The query to the $j$-th server is written as $Q_{\tau^{-1}(j)} \in \Fq^{|\tau^{-1}(j)|\times \sff\sx }$.

%
%
%
%
%
\end{description}

A projected linear SPIR is described by the triplet $(\sH, \sJ, \tau)$ of randomness encoder $\sH\in\Fq^{\sz\times \sy}$, user's query encoder $\sJ\in \Fq^{\sz\times \sx}$, and position map $\tau: [\sz] \to [\snn]$.

\end{defi}

Whereas the size of the query $Q\in\Fq^{\sz\times \sff\sx}$ in linear SPIR is increasing with the number of files $\sff$,
    that of projected linear SPIR $\sJ\in \Fq^{\sz\times \sx}$ is smaller and independent of the size of the files.

We define a projection $\hat{p}$ from linear SPIR protocols into projected linear SPIR protocols as follows.

\begin{defi}[Projection] \label{def:projection1}
For a linear SPIR protocol $\PSPIR$, let $\sQ = (\sQ',\sQ'')\in \Fq^{\sz\times \sx} \times \Fq^{\sz\times (\sff-1)\sx}$ be the matrix of query when $K=1$ and $R=0$.
Further, let $\sH\in \Fq^{\sz\times \sy} $ be the randomness encoder and $\tau$ be the position map of $\PSPIR$.
Then, 
    we define a map $\hat{p}$ that maps $\PSPIR$
        to the projected linear SPIR protocol described by the triplet $(\sH, \sQ', \tau)$ of the randomness encoder, user's query encoder, and position map.
The map $\hat{p}$ is a projection to projected linear SPIR because any projected linear SPIR protocol is mapped to itself.
\end{defi}


\subsection{Formal definition of non-perfect secret sharing (NSS)}



We formally define an NSS protocol with one dealer and $\snn$ players as follows. 
\begin{defi}[$(\ACC,\REJ,\snn)$-NSS] \label{def:nss}

A secret is given as a uniform random variable $L \in \cM$ and $\smm \coloneqq |\cM|$.
A protocol $\PNSS$ is defined by 
    the following deterministic map (the encoding map) $f$ to generate n shares in the share generation step;
\begin{itemize}[leftmargin=1.5em]
\item \textbf{Share generation}:
Depending on the dealer's private uniform randomness $\RandNSS\in\SRandNSS$,
	the dealer prepares $\snn$ shares by 
\begin{align}
S = (S_1,\ldots, S_{\snn}) = f(L,\RandNSS) \in \cS_1\times \cdots \times \cS_{\snn},
\end{align}
and sends the $j$-th share $S_j$ to the $j$-th player.

\end{itemize}

For an access structure $(\ACC, \REJ)$  on $[\snn]$,
	the protocol $\PNSS$ is called an $(\ACC,\REJ,\snn)$-NSS protocol with security $(\alpha ,\beta )$ if 
	the following conditions are satisfied.
\begin{itemize}[leftmargin=1.5em]
\item \textbf{Correctness}:
We define the maximum likelihood decoder 
    \begin{align*}
    \hat{\ell}^{\mathrm{ML}}_{s_{\cA}} &\coloneqq \argmax_{\ell\in\cM} \Pr[S_{\cA} = s_{\cA}| L = \ell].
    \end{align*}
Then, 
\begin{align}
\alpha \ge 
\alpha (\PNSS)\coloneqq \max_{\cA\in\fA} \pr_{L,S_{\cA}}[L \neq \hat{\ell}^{\mathrm{ML}}_{S_{\cA}}].
\label{correct:spir}
\end{align}


\item \textbf{Secrecy}:
\begin{align}
	\beta \ge \beta(\PNSS) \coloneqq
	\max_{\cB\in \REJ} I(L ;S_{\cB}) .
\end{align}

\end{itemize}
The SS rate of the protocol $\PNSS$ is defined as 
	\begin{align}
	\RSS(\PNSS) \coloneqq
		\frac{\log \smm}{ \sum_{j=1}^{\snn} \log |\cS_j|}.
	\end{align}
The randomness rate is defined as 
\begin{align}
R_{\RandNSS} = \frac{\log |\SRandNSS|}{\log \smm}.
\end{align}
\end{defi}

In the above security conditions,
$\alpha (\PSPIR)$ is the worst-case error probability with the maximum likelihood decoder and
$\beta (\PSPIR)$ is the worst-case leakage of the secret to the forbidden set of shares.
If $(\alpha,\beta) = (0,0)$, 
	the $(\ACC,\REJ,\snn)$-CSNSS protocol has complete security.
\begin{defi}[CSNSS]
A $(\ACC,\REJ,\snn)$-NSS protocol with security $(\alpha,\beta) = (0,0)$ is called 
    a completely secure $(\ACC,\REJ,\snn)$-NSS ($(\ACC,\REJ,\snn)$-CSNSS) protocol.
\end{defi}
In a $(\ACC,\REJ,\snn)$-CSNSS protocol, 
		the players indexed by any $\cA\in\ACC$ can recover $L$ without error, 
	but the players indexed by any $\cB \in \REJ$ obtain no information of $L$.
The achievable rate and capacity of $(\ACC,\REJ,\snn)$-CSNSS are defined as follows.

\begin{defi}[Achievability of $(\ACC,\REJ,\snn)$-CSNSS]
A SS rate $R$ is achievable if
	there is a sequence of $(\ACC,\REJ,\snn)$-CSNSS protocols 
		$\{ \PNSS \}_{\smm\in\mathbb{N}}$
	such that $\RSS(\PNSS) \to R$ as  $\smm\to\infty$.
\end{defi}

\begin{defi}[Capacity of $(\ACC,\REJ,\snn)$-CSNSS]
Capacity of $(\ACC,\REJ,\snn)$-CSNSS, 
	$C^{(\REC,\COL,\snn)}_{\mathrm{NSS}}$,
	is the supremum of 
	achievable rates of $(\ACC,\REJ,\snn)$-CSNSS.
%
\end{defi}

\begin{remark}
When $\fA \cup \fB = 2^{[\snn]}$, the protocol in Definition~\ref{def:nss} is called a {\em (perfect) secret sharing} protocol.
Since all results in this paper are obtained for any general access structure, 
    all results also apply to perfect SS.
\end{remark}

As special cases, NSS with thresholds and linear NSS are defined as follows.

\begin{defi}[$(\srr,\stt,\snn)$-NSS]
When $\ACC = \{ \cA \subset [\snn] \mid |\cA| \geq \srr \}$
	and $\REJ = \{ \cB \subset [\snn] \mid |\cB| \leq \stt \}$,
	$(\ACC,\REJ,\snn)$-NSS protocols are called $(\srr,\stt,\snn)$-NSS protocols.
\end{defi}

\begin{defi}[Linear NSS]
A protocol $\PNSS$ is called a {\em linear NSS} protocol if 
    the encoding map $f$ satisfies the following conditions with a position map $\tau: [\sz] \to [\snn]$.
    \begin{description}
    \item[Vector representation of secret]
    The secret $L$ is written as a vector in $\FF_q^{\sx}$.
	\item[Vector representation of randomness]
	The dealer's private randomness $\RandNSS$ is written as a uniform random vector in $\FF_q^{\sy}$.
	\item[Linearity of share generation]
	The encoder $f$ is written as a linear map from $\FF_q^{\sx+\sy}$ to $\FF_q^{\sz}$ that maps
	$(Z_1,\ldots, Z_\sz) = f (L, \RandNSS)$.
	The $j$-th secret is written as $S_j =  ( Z_{i} \mid i \in \tau^{-1}(j) ) \in \Fq^{|\tau^{-1}(j)|}$.
	\end{description}
	Thus, a linear NSS $\PNSS$ is described by the pair $(f, \tau)$ of a linear map $f:\Fq^{\sx+\sy}\to\Fq^{\sz}$ and a position map $\tau: [\sz] \to [\snn]$.  
	Due to the above conditions, the SS rate and randomness rate of the linear NSS are $\sx/\sz$ and $\sy/\sx$, respectively.
\end{defi}

\subsection{Formal definition of multi-target monotone span program (MMSP)}

We define  a multi-target monotone program (MMSP) 
	with an access structure as follows.
\begin{defi}[{Multi-target monotone span program (MMSP)}] \label{defi:MMSP}
Given positive integers $\snn,\sx,\sy,\sz$, 
	a pair $\cP = (\sG ,\tau)$
	of 
	a matrix $\sG= (\sG',\sG'') \in \FF_q^{\sz\times (\sx+\sy)}  = \FF_q^{\sz\times \sx} \times \FF_q^{\sz\times \sy}$
	and a position map $\tau: [\sz] \to [\snn]$
	is called a multi-target monotone span program (MMSP).


Let $\mathbf{e}_i\in\Fq^{\sx+\sy}$ be the row vector with $1$ in the $i$-th coordinate and $0$ in the others.
Let $\cE \coloneqq \{\mathbf{e}_1,\ldots, \mathbf{e}_\sx\}$ and
    we say the following;
\begin{itemize}
\item $\cP$ accepts $\ACC$ if 
		$\cE$ is included in the row space of $\sG_{\tau^{-1}(\cA)}$
			for any $\cA\in\ACC$, and
		
\item $\cP$ rejects $\REJ$ if 
        $\spann \cE \cap \rowspan \sG_{\tau^{-1}(\cB)} = \{0 \}$ for any $\cB\in\REJ$.
\end{itemize}
An MMSP $\cP$ is called $(\ACC,\REJ,\snn)$-MMSP if $\cP$ accepts $\ACC$ and rejects $\REJ$.
The MMSP rate $\RMMSP$ is defined as the ratio $\sx/\sz$.
\end{defi}

\begin{remark}
Our definition of MMSP generalizes the definition in \cite{Beimel11,BI93,Dijk95}.
The MMSP defined in \cite{Beimel11,BI93,Dijk95} corresponds to our definition with 
    $\ACC \cup \REJ = 2^{[\snn]}$ and $\ACC \cap \REJ = \emptyset$, i.e., 
	every subset of $[\snn]$ is either authorized or forbidden.
Our definition of MMSP also generalizes the monotone span programs \cite{KW93},
    which corresponds to the case $\sx = 1$ and $\ACC\cup\REJ = 2^{[\snn]}$ for our MMSP definition.

The papers \cite{Brickell89, KW93,Beimel_thesis} proved 
	the equivalence of linear perfect SS protocols with complete security and monotone span programs.
This paper generalizes this relation to the equivalence of linear CSNSS protocols and MMSPs.
\end{remark}

As special cases, we define $(\ACC,\REJ,\snn)$-MMSPs with thresholds as follows.

\begin{defi}[$(\srr,\stt,\snn)$-MMSP]
When $\ACC = \{ \cA \subset [\snn] \mid |\cA| \geq \srr \}$
	and $\REJ = \{ \cB \subset [\snn] \mid |\cB| \leq \stt \}$,
	$(\ACC,\REJ,\snn)$-MMSPs are called $(\srr,\stt,\snn)$-MMSPs.
\end{defi}

$(\srr,\stt,\snn)$-MMSPs are related with maximum distance separable (MDS) codes.
The column space of a matrix $\sA\in\FF_q^{\snn\times\skk}$ is called an $(\snn,\skk)$-MDS code
	if any $\skk$ rows of $\sA$ are linearly independent.
We give the relation between $(\srr,\stt,\snn)$-MMSPs and MDS codes by the following theorem. 

\begin{theo} \label{theo:MDS-MMSP}
Let 
$\sG = (\sG', \sG'') \in \FF_q^{\snn\times (\srr-\stt)} \times \FF_q^{\snn\times \stt}$.
The following conditions are equivalent.

\begin{enumerate}[label=(\alph*)]
\item With $\tau \coloneqq \id_{[\snn]}$ and $(\sx,\sy,\sz) \coloneqq (\srr-\stt, \stt, \snn)$, $\cP = (\sG,\tau)$ is an $(\srr,\stt,\snn)$-MMSP. 
\item $\sG$ ($\sG''$) is the generator matrix of an $(\snn,\srr)$-MDS code ($(\snn,\stt)$-MDS code).
\end{enumerate}
\end{theo}

Theorem~\ref{theo:MDS-MMSP} will be proved in Appendix~\ref{MMSP-MDS}.

The matrices $\sG$ and $\sG''$ with condition (b) of Theorem~\ref{theo:MDS-MMSP}
 have been applied for 
    private information retrieval \cite{WS17,WS17-2, BU18, FHGHK17, TGKFH19}, 
	secret sharing \cite{Shamir79, Yamamoto86, BK95, CCGHV07, DP10, KMU12, HB17}, 
	wiretap channel II \cite{OW84,SM09},
	secure network coding \cite{CY11, CH17, ERSS12,NYZ11, SK11, ZZ09, KMU15},
    distributed storage system \cite{DGWR10}
	and cryptography \cite{Vaudenay, ZSE02}.
Especially, Ozarow and Wyner's optimal wiretap channel code corresponds to the case $\srr = \snn$ and Shamir's secret sharing protocol corresponds to the case where $\srr = \stt + 1$.
Thus, from Theorem~\ref{theo:MDS-MMSP}, $(\srr,\stt,\snn)$-MMSP with condition (a) characterizes the structure of those protocols
    and $(\ACC,\REJ,\snn)$-MMSP is a generalization of those protocols with general access structures.




\section{Main Results} \label{sec:main_result}

In this section, we present our main theorems in three subsections.
The subsections present, respectively, the conversions from SPIR to NSS, from linear CSNSS to MMSP, and from MMSP to projected linear CSSPIR.

\subsection{Conversion from SPIR to NSS} \label{sec:main1}


The following protocol converts a SPIR protocol to an NSS protocol. 
\begin{prot} \label{prot:SPIRtoSS}
Given a SPIR protocol $\PSPIR$,
    an NSS protocol $\hPNSS[\PSPIR]$ is constructed by generating $\snn$ shares as follows.
Recall that from the definition of the SPIR protocol $\PSPIR$, the symbols $\cM$, $\SRandSPIR$, and $\cR$ denote the space of files, random seeds for shared randomness, and private randomnesses of the user, respectively.

We define $r^{\ast}\in\cR$ and $\mast\in \cM^{\sff-1}$ as
    \begin{align}
    &r^{\ast} \in \argmin_{r\in\cR}  \px*{ \max_{\cB \in \REJ} I( M_{1} ; D_{\cB} | R=r, K=1) },  \label{lemmeq:correct} \\
    &\mast \in \argmin_{m_{[2:\sff]}\in \cM^{\sff-1}}  I(M_1 ; D_{\cB} |M_{[2:\sff]}=m_{[2:\sff]}, R=r^{\ast}, K=1).
    \label{eq:m2nast}
    \end{align}
The values $r^\ast$ and $\mast$ are publicly known.

The dealer chooses the secret $L\in\cM$ and the uniform randomness $\RandNSS \in \SRandSPIR$.
%
%
%
%
%
The dealer simulates $\PSPIR$
    with $K\coloneqq1$, 
    $M_1 \coloneqq L$, 
    and $\RandSPIR\coloneqq \RandNSS $
    while fixing the user's private randomness as $R \coloneqq r^{\ast}$
    and the 2nd, \ldots, $\sff$-th file as $M_{[2:\sff]} \coloneqq \mast$.
From this simulation of $\PSPIR$, the dealer generates the answers $D_1,\ldots, D_\snn$ of the SPIR protocol $\PSPIR$
    and sets the $j$-th share for NSS as $S_j \coloneqq D_j$. 
\end{prot}

Our first result is as follows. 
\begin{theo} \label{theo:SPIRtoNSS}
Let $\PSPIR$ be an $(\REC,\COL,\snn,\sff)$-SPIR protocol with 
rate $\RPIR$,
shared randomness rate $R_{\RandSPIR}$,
security $(\alpha,\beta, \gamma)$.
Then, the NSS protocol $\hPNSS[\PSPIR]$ defined in Protocol~\ref{prot:SPIRtoSS}
	is an $(\ACC,\REJ,\snn)$-NSS protocol $\PNSS$ with 
	SS rate $\RSS = \RPIR$, 
	randomness rate $R_{\RandNSS} = R_{\RandSPIR}$,
	and security $(\alpha, \xi(\alpha,\beta,\gamma))$, where 
	\begin{align}
	\xi(\alpha,\beta,\gamma) 
	&= 2 \beta +  (1-\alpha+4\sqrt{2\gamma\sff}) \log \smm \\
	&\quad + 2h_2(\sqrt{2\gamma\sff})
+ h_2(1-\alpha) + \log \alpha
    \end{align}
	and $h_2$ is the binary entropy function $h_2(p) =  -p \log p - (1-p)\log (1-p)$.
	Here, $\xi(\alpha,\beta,\gamma)$ goes to $0$ as $(\alpha,\beta,\gamma) \to (0,0,0)$.
\end{theo}
We give the proof of Theorem~\ref{theo:SPIRtoNSS} in Section~\ref{sec:SPIRtoNSS}.

For the case of complete security, i.e., $(\alpha,\beta,\gamma) = (0,0,0)$, Theorem~\ref{theo:SPIRtoNSS} implies that Protocol~\ref{prot:SPIRtoSS} from CSSPIR is CSNSS.
Furthermore, Protocol~\ref{prot:SPIRtoSS} converts a linear CSSPIR protocol into a linear CSNSS protocol as follows.

\begin{coro}	\label{coro:linearSPIRtoNSS}
Let $\PSPIR$ be a {\em linear} $(\REC,\COL,\snn,\sff)$-CSSPIR protocol with rate $\RPIR$, shared randomness rate $R_{\RandSPIR}$, and position map $\tau$.
Then, the NSS protocol $\hPNSS[\PSPIR]$ defined in Protocol~\ref{prot:SPIRtoSS} with $r^\ast = 0$, $\mast = 0$
	is a {\em linear} $(\ACC,\REJ,\snn)$-CSNSS protocol with SS rate $\RSS = \RPIR$, randomness rate $R_{\RandNSS} = R_{\RandSPIR}$, and position map $\tau$.
\end{coro}



We give the proof of Corollary~\ref{coro:linearSPIRtoNSS} in Section~\ref{sec:SPIRtoNSS}.

In Corollary~\ref{coro:linearSPIRtoNSS}, 
    the values of $r^\ast$ and $\mast$ is chosen as zero instead of \eqref{lemmeq:correct} and \eqref{eq:m2nast}.
This choice is justified by the following lemma.

\begin{lemm} \label{lemm:seccrecty}
For any $(\REC,\COL,\snn,\sff)$-CSSPIR, any $k\in[\sff]$, any $r\in\cR$, and any $\cB\in\fB$, we have 
\begin{align*}
 I(M_{[\snn]} ; D_{\cB} |  K=k, R=r) 
&= 0.
\end{align*}
\end{lemm}
\begin{proof}
From the chain rule of the mutual information, we have 
\begin{align}
& I(M_{[\snn]} ; D_{\cB} |  K=k, R=r)  \nonumber \\
&= I(M_k ; D_{\cB} | K=k, R=r) \nonumber \\
   &\quad + 
    I(M_{[\snn]\setminus\{k\}} ; D_{\cB} | M_k, K=k, R=r)
    \label{eq:asdfseeee}
\end{align}
The second term of \eqref{eq:asdfseeee} is zero from 
\begin{align}
&I(M_{[\snn]\setminus\{k\}} ; D_{\cB} | M_k, K=k, R=r)\\
&= I(M_{[\snn]\setminus\{k\}} ; D_{\cB} | K=k, R=r)
= 0 ,
\end{align}
where the first equality follows from the correctness condition and the second equality is from the server secrecy of the CSSPIR protocol.
The first term of \eqref{eq:asdfseeee} is zero as follows.
From the user secrecy of the CSSPIR protocol,
    the answers $D_{\cB}$ are generated independently of $K$,
    which implies that 
    \begin{align*}
& I(M_k ; D_{\cB} | K=k, R=r) 
= I(M_k ; D_{\cB} | K=k', R=r)
= 0,
\end{align*}
where the last equality follows from the server secrecy.
Thus, we obtain the desired statement.
%
\end{proof}

Next, we give the proof idea of Theorem~\ref{theo:SPIRtoNSS} with Lemma~\ref{lemm:seccrecty}.
If we only consider the case of complete security, 
    Theorem~\ref{theo:SPIRtoNSS} for CSSPIR is simply proved as follows.
First, we prove the complete correctness of the induced NSS protocol.
For $\cA \in \fA$, the first file $M_1$ is recovered from $D_{\cA}$ since the answers $D_{\cA}$ are generated from the SPIR protocol $\PSPIR$ for the case $K=1$, and the randomness $R=r^{\ast}$ is publicly known.
Thus, the secret $L = M_1$ of the NSS protocol $\hPNSS[\PSPIR]$ is recovered from  the shares $S_{\cA} = D_{\cA}$.
Next, we prove the complete secrecy against forbidden sets $\cB \in \fB$. 
From Lemma~\ref{lemm:seccrecty}, 
    we have $I(M_{1} ; D_{\cB} | M_{[2:\sff]}=\mast, K=1, R=r^\ast) = 0$,
    which implies that the shares $S_{\cB} = D_{\cB}$ have no information of the secret $L = M_1$.
Thus, Theorem~\ref{theo:SPIRtoNSS} for CSSPIR is proved.




Now, we discuss the achievable rate and capacity of CSSPIR.
From Theorem~\ref{theo:SPIRtoNSS} for $(\alpha,\beta,\gamma) = (0,0,0)$, we obtain the following corollary.
\begin{coro}	\label{coro:SPIRtoNSS}
If there is an $(\REC,\COL,\snn,\sff)$-CSSPIR protocol with PIR rate $\RPIR$ and shared randomness rate $R_{\RandSPIR}$,
	the SS rate $\RPIR$ is achievable for $(\REC,\COL,\snn)$-NSS with randomness rate $R_{\RandNSS} = R_{\RandSPIR}$.
\end{coro}

An upper bound of SS rate for $(\ACC,\REJ,\snn)$-CSNSS is proved in \cite{OKT93,OK95,Paillier98} as follows.  
\begin{prop}[{\cite{OKT93,OK95,Paillier98}}] \label{prop:rssup}
For any $\ACC,\REJ\subset 2^{[\snn]}$, let $\delta(\ACC,\REJ) \coloneqq \min \{ |\cA-\cB| \mid \cA\in \ACC, \ \cB\in \REJ \}$.
Any $(\ACC,\REJ,\snn)$-CSNSS protocol $\PNSS$ satisfies
\begin{align}
\RSS(\PNSS)
	\leq \frac{1}{\snn} \delta(\ACC,\REJ).
\end{align}
In particular, for $(\srr,\stt,\snn)$-CSNSS, $\RSS \leq \srr-\stt$.
\end{prop}

As a corollary of Theorem~\ref{theo:SPIRtoNSS} and Proposition~\ref{prop:rssup}, we obtain an upper bound of the CSSPIR capacity.
\begin{coro} \label{coro:capacity}
For any $\ACC,\REJ\subset 2^{[\snn]}$ and $\delta(\ACC,\REJ)$ defined in Proposition~\ref{prop:rssup}, 
    we have
$$C^{(\REC,\COL,\snn,\sff)}_{\mathrm{SPIR}} 
			\leq  C^{(\REC,\COL,\snn)}_{\mathrm{NSS}} 
			\leq \frac{1}{\snn}\delta(\ACC,\REJ).$$
\end{coro}

Corollary~\ref{coro:capacity} is applicable for any access structure and is simply characterized by the access structure.
For example, when a CSSPIR protocol has the access structure $\snn = 3$, $\fA = \{  \{2,3\}, \{1,2,3\} \}$, and $\fB = \{\emptyset,  \{1\},  \{2\}, \{3\}, \{1,2\} \}$ given in Example~\ref{exam:as}, Corollary~\ref{coro:capacity} implies that the SPIR rate is upper bounded by $1/3$ because $\snn = 3$ and $\delta(\fA,\fB) = |\{2,3\} - \{2\} | = 1$.

Furthermore, the upper bound in Corollary~\ref{coro:capacity} is tight for the threshold case as follows.

\begin{coro} \label{coro:cap}
The capacity of $\snn$-server CSSPIR with $\srr$ responsive servers and $\stt$ colluding servers is $(\srr-\stt)/\snn$.
\end{coro}

The converse part is proved by Corollary~\ref{coro:capacity}.
The achievability part of Corollary~\ref{coro:cap} follows from the linear CSSPIR protocol of Tajeddine et al. \cite{TGKFH19}.
Tajeddine et al. \cite{TGKFH19} constructed a protocol of symmetric/non-symmetric CSSPIR from coded storage with colluding, byzantine, and unresponsive servers, with PIR rate 
\begin{align}
    \frac{\srr-\skk-2\sbb-\stt+1}{\snn}, 
 \label{eq:tajedd}
\end{align}
    where
	$\srr$ is the number of responding servers,
	$\skk$ is the code rate of the coded storage,
	$\sbb$ is the number of byzantine servers.
When $(\skk,\sbb) = (1,0)$, their protocol is an $(\srr,\stt,\snn,\sff)$-CSSPIR protocol and achieves the PIR rate $(\srr-\stt)/\snn$.



Holzbaur et al. \cite[Theorem 4]{HFLH21} proved that the rate \eqref{eq:tajedd} is optimal for linear CSSPIR with additive randomness.
Thus, when $(\skk,\sbb) = (1,0)$, the capacity for linear CSSPIR with additive randomness is the same as Corollary~\ref{coro:cap}.
Our result generalizes this result because Corollary~\ref{coro:cap} holds without the assumptions of the linearity of protocols and the additivity of the randomness.

\begin{remark}
%
%
The same implication from SPIR to NSS of Theorem~\ref{theo:SPIRtoNSS} is applicable for multi-round SPIR by the same proof.
Especially, since the multi-round SPIR capacity is greater than the one-round SPIR capacity, 
	our result also implies that the multi-round capacity of $(\srr,\stt,\snn)$-CSSPIR is $(\srr-\stt)/\snn$ which is the same as the one-round capacity.
Moreover, Theorem~\ref{theo:SPIRtoNSS} is also applicable even for multi-round SPIR with coded database.
However, for simplicity, this paper only discusses one-round protocols when all files replicated in each server. 
\end{remark}

\subsection{Conversion from linear CSNSS to MMSP} \label{sec:main3}

In this subsection, 
	we give the conversion from linear CSNSS to MMSP.

\begin{prot} \label{prot:NSStoMMSP}
Let $\PNSS$ be a linear CSNSS protocol defined from a linear encoder $f:\FF_q^{\sx+\sy}\to\FF_q^{\sz}$ and a map $\tau:[\sz]\to[\snn]$.
Let $\sG_f$ be the matrix representation of the linear map $f$.
Then, an MMSP $\hPMMSP[\PNSS] \coloneqq (\sG_f, \tau)$ is defined.
\end{prot}


\begin{theo}	\label{theo:NSS=MMSP}
Given a linear $(\ACC,\REJ,\snn)$-CSNSS protocol $\PNSS$ with SS rate $\RSS$,
	the MMSP $\hPMMSP[\PNSS]$ defined in Protocol~\ref{prot:NSStoMMSP} 
		is an $(\ACC,\REJ,\snn)$-MMSP with MMSP rate $\RMMSP = \RSS$.
\end{theo}

We give the proof of Theorem~\ref{theo:NSS=MMSP} in Section~\ref{sec:NSS=MMSP}.



\subsection{Conversion from MMSP to projected linear CSSPIR} \label{sec:main2}


In this subsection, we give the conversion from MMSP to projected linear CSSPIR.
First, we define a projected linear SPIR protocol from MMSP. 

\begin{prot} \label{prot:MMSPtoSPIR}
Let $\cP  = (\sG, \tau)$ be an MMSP. 
We denote 
	$\sG = (\sG', \sG'') \in \FF_q^{\sz\times \sx} \times \FF_q^{\sz\times \sy}.$
Let $\sff$ be an arbitrary positive integer at least $2$.
Then, a projected linear CSSPIR protocol $\hPSPIR[\cP]$ is defined by the triplet $(\sG'', \sG', \tau)$ of randomness encoder $\sG''\in\Fq^{\sz\times \sy}$, user's query encoder $\sG'\in \Fq^{\sz\times \sx }$, and position map $\tau: [\sz] \to [\snn]$.

\end{prot}


The SPIR protocol defined in Protocol~\ref{prot:MMSPtoSPIR} is completely secure by the following theorem.

\begin{theo}	\label{theo:MMSPtoSPIR}
Let $\cP$ be an $(\ACC,\REJ,\snn)$-MMSP with 
a matrix $\sG \in \Fq^{\sz \times (\sx+\sy)}$,
a position map $\tau:[\sz]\to[\sx]$,
and the rate $\RMMSP=\sx/\sz$.
Then, the SPIR protocol $\hPSPIR[\cP]$ defined in Protocol~\ref{prot:MMSPtoSPIR}
	is 
	an $(\REC,\COL,\snn,\sff)$-CSSPIR protocol $\PSPIR$ for any $\sff \geq 2$ with PIR rate $\RPIR = \RMMSP= \sx/\sz$ and shared randomness rate $\sy/\sx$.
\end{theo}

We give the proof of Theorem~\ref{theo:MMSPtoSPIR} in Section~\ref{sec:NSStoSPIR}.

\subsection{Equivalence of linear CSSPIR, linear CSNSS, and MMSP} \label{sec:equiv_subsec}

Combining Corollary~\ref{coro:linearSPIRtoNSS}, Theorem~\ref{theo:NSS=MMSP}, and Theorem~\ref{theo:MMSPtoSPIR},
	we obtain the equivalence of linear CSSPIR, linear CSNSS, and MMSP.

\begin{coro}	\label{coro:equiv}
Let
    $\snn$ be a positive integer at least $2$, 
    $(\ACC,\REJ)$ be an access structure on $[\snn]$, 
    and 
    $\tau$ be a map from $[\sz]$ to $[\snn]$.

The following conditions are equivalent.
\begin{enumerate}[label=(\alph*)]
\item For some $\sff \geq 2$, there exists a linear $(\ACC,\REJ,\snn,\sff)$-CSSPIR protocol with 
    the position map $\tau$,
    the rate $\RPIR = \sx/\sz$, 
    and 
    the shared randomness rate $R_{\RandSPIR} = \sy/\sx$.

\item There exists a linear $(\ACC,\REJ,\snn)$-CSNSS protocol with 
    the position map $\tau$,
    the rate $\RSS = \sx/\sz$,
    and 
    the randomness rate $R_{\RandNSS} = \sy/\sx$.

\item There exists a $(\ACC,\REJ,\snn)$-MMSP with 
    the matrix $\sG \in \Fq^{\sz\times (\sx+\sy)}$
    the position map $\tau$,
    and 
    rate $\RMMSP  = \sx/\sz$.

\item For any $\sff \geq 2$, there exists a projected linear $(\ACC,\REJ,\snn,\sff)$-CSSPIR protocol 
with 
    the position map $\tau$,
    the rate $\RPIR = \sx/\sz$, 
    and 
    the shared randomness rate $R_{\RandSPIR} = \sy/\sx$.
\end{enumerate}
\end{coro}

Corollary~\ref{coro:equiv} is proved as follows.
The relations $(a)\Rightarrow(b)$, $(b)\Rightarrow(c)$, and $(c)\Rightarrow(d)$ follow from Corollary~\ref{coro:linearSPIRtoNSS}, Theorem~\ref{theo:NSS=MMSP}, and Theorem~\ref{theo:MMSPtoSPIR}, respectively,
    and $(d)\Rightarrow(a)$ is trivial.


For any linear CSNSS protocol $\PNSS$, we have
\begin{align}
\hPNSS[\hPSPIR[\hPMMSP[\PNSS]]]=\PNSS.
\end{align}
That is, the composite map $\hPNSS\circ \hPSPIR \circ \hPMMSP$ is the identity map on linear NSS.
Similarly, the composite map $\hPSPIR\circ \hPMMSP\circ \hPNSS$ is 
    the projection $\hat{p}$ into projected linear SPIR, defined in Definition~\ref{def:projection1}, i.e.,
    for any a projected linear CSSPIR protocol $\PSPIR$, we have 
\begin{align}
\hPSPIR[\hPMMSP[\hPNSS[\PSPIR]]]= \hat{p}(\PSPIR) = \PSPIR.
\end{align}


A linear CSSPIR protocol $\PSPIR$ and its projected protocol $\hat{p}(\PSPIR)$ are compared as follows.
In the two protocols,
    the sizes of queries, answers, and random seed are the same. 
However, 
    from definitions of linear SPIR and projected linear SPIR,
the projected protocol is more concisely described. 
That is, 
    whereas the random query matrices $Q \in \Fq^{\sz\times \sff\sx}$ are necessary to be characterized in the original CSSPIR protocol $\PSPIR$,
    all queries of the projected protocol  $\hat{p}(\PSPIR)$ can be described only with a small matrix $\sJ\in\Fq^{\sz\times \sx}$.
On the other hand, 
    since the size of user's private randomness $R\in\Fq^{\sy\times \sff\sx}$ in the projected protocol is increasing with the number of the files $\sff$, 
    one possible advantage of linear CSSPIR protocols would be the smaller size of this randomness.

Next, we discuss the relation between projected linear CSSPIR and linear CSNSS as follows.
In both protocols, the minimum information to describe encoders is 
    two matrices in $\Fq^{\sz\times \sx}$ and $\Fq^{\sz\times \sy}$.
The matrix in $\Fq^{\sz\times \sx}$ 
    is used for encoding messages (targeted file for CSSPIR and the secret for CSNSS)
    and
    the matrix in $\Fq^{\sz\times \sy}$ for encoding randomness (and non-targeted files for CSSPIR).
From this similarity,
    the server secrecy of CSSPIR and the secrecy of CSNSS are guaranteed in the same context.
On the other hand, the uniform randomness for CSNSS is the dealer's randomness $\RandNSS\in\Fq^{\sy}$,
    while that for CSSPIR is 
    the random seed $\RandSPIR\in\Fq^{\sy}$ and the user's private randomness $R\in\Fq^{\sy\times\sff\sx}$.
In CSSPIR, the additional randomness $R$ is required for guaranteeing the user secrecy.

Farr\`as et al. proved the existence of a linear CSNSS protocol for any access structure \cite{FHKP17}.
\begin{prop}[{\cite{FHKP17}}] \label{prop:anyNSS}
For any access structure $(\fA,\fB)$,
    there exists a linear $(\REC,\COL,\snn)$-CSNSS protocol.
\end{prop}

Combining 
    Proposition~\ref{prop:anyNSS}
    and
    Corollary~\ref{coro:equiv},
    we obtain the following corollary.

\begin{coro}
For any access structure $(\ACC,\REJ)$ on $[\snn]$ and any $\sff\geq 2$, there exists an $(\REC,\COL,\snn,\sff)$-CSSPIR protocol.
\end{coro}

\section{Proof of Conversion from SPIR to NSS} \label{sec:SPIRtoNSS}

In this section, we prove Theorem~\ref{theo:SPIRtoNSS} and Corollary~\ref{coro:SPIRtoNSS}.

\subsection{Proof of Theorem~\ref{theo:SPIRtoNSS}}
For the proof, we prepare the following lemma, which extends Lemma~\ref{lemm:seccrecty} to the incomplete secrecy case.

%

\begin{lemm}	\label{lemm:noleak}
For any $(\fA,\fB,\snn)$-SPIR protocol $\PSPIR$, any $\cB\in\fB$, and $r^{\ast}$ defined in \eqref{lemmeq:correct},
    we have
\begin{align}
&I(M_1 ; D_{\cB} |M_{[2:\sff]}, R=r^{\ast}, K=1) \\
&\leq 2 \beta +  (1-\alpha+4\sqrt{2\gamma\sff}) \log \smm + 2h_2(\sqrt{2\gamma\sff})
+ h_2(1-\alpha) + \log \alpha.
\end{align}
    
\end{lemm}

The proof of Lemma~\ref{lemm:noleak} is given in Appendix~\ref{append:noleak}.
Now, we prove Theorem~\ref{theo:SPIRtoNSS}. 

\begin{proof}[Proof of Theorem~\ref{theo:SPIRtoNSS}]
The correctness of the NSS protocol $\hPNSS[\PSPIR]$ is upper bounded by $\alpha$ because 
    we have
    \begin{align}
    \alpha(\hPNSS[\PSPIR]) &\stackrel{\mathclap{(a)}}{=} 
        \max_{\cA\in\fA} \pr_{M_1,D_{\cA}} [M_1 \neq \hat{m}^{\mathrm{ML}}_{D_{\cA},r^{\ast},k=1}] \\
         &\stackrel{\mathclap{(b)}}{\leq} \alpha(\PSPIR),
    \end{align}
    where $(a)$ follows from the definition of Protocol~\ref{prot:SPIRtoSS} and $(b)$ follows the correctness condition of the SPIR protocol $\PSPIR$.
The secrecy $\beta(\hPNSS[\PSPIR])$ is upper bounded as
    \begin{align}
    &I(M_1 ; D_{\cB} |M_{[2:\sff]}=\mast, R=r^{\ast}, K=1) \\
    &\stackrel{\mathclap{(a)}}{\leq} I(M_1 ; D_{\cB} |M_{[2:\sff]}, R=r^{\ast}, K=1) \\
    &\stackrel{\mathclap{(b)}}{\leq} 2 \beta +  (1-\alpha+4\sqrt{2\gamma\sff}) \log \smm + 2h_2(\sqrt{2\gamma\sff}) + h_2(1-\alpha) + \log \alpha \\ 
    & = \xi(\alpha,\beta,\gamma),
    \end{align}
    where 
        $(a)$ follows from the choice of $\mast$ in \eqref{eq:m2nast} and
        $(b)$ follows from Lemma~\ref{lemm:noleak}.
Thus, the protocol $\hPNSS[\PSPIR]$ is an $(\REC,\COL,\snn,\sff)$-NSS protocol with the desired security parameters $(\alpha,\xi(\alpha,\beta,\gamma))$.
The SS rate of Protocol~\ref{prot:SPIRtoSS} is 
	\begin{align}
	&\RSS =  
		\frac{\snn\log \smm}{\sum_{j\in[\snn]} \log |\cS_j|}
		= 
		\frac{\snn\log \smm}{\sum_{j\in[\snn]} \log |\cD_j|}
		= \RPIR,
		\label{eq:label2}
	\end{align}
	which proves Theorem~\ref{theo:SPIRtoNSS}. 
\end{proof}

\subsection{Proof of Theorem~\ref{coro:SPIRtoNSS}}

The proof of Corollary~\ref{coro:linearSPIRtoNSS} is as follows.
The conditions for the access structure, security, and rate follow from Theorem~\ref{theo:SPIRtoNSS}.
Thus, it is enough to prove that the NSS protocol $\hPNSS[\PSPIR]$ is linear. 
To prove the linearity, we first analyze  the simulation of the linear SPIR protocol $\PSPIR$ in Protocol~\ref{prot:SPIRtoSS}, and then,
    we prove that the NSS protocol $\hPNSS[\PSPIR]$ is linear.

First, we analyze Protocol~\ref{prot:SPIRtoSS} with parameters defined in Corollary~\ref{coro:linearSPIRtoNSS}.
In Corollary~\ref{coro:linearSPIRtoNSS}, the linear CSSPIR protocol $\PSPIR$ is simulated with $K=1$, $R=0$, and $M_{[2:\sff]}=0$.
Since the query is determined by $K$ and $R$, the simulated query is fixed as a matrix $\sQ \in \Fq^{\sz\times \sff\sx}$.
We denote $\sQ = (\sQ', \sQ'') \in \Fq^{\sz \times \sx} \times \Fq^{\sz\times (\sff-1)\sx}$.
With the uniform randomness $\RandSPIR\in\Fq^{\sy}$ and the randomness encoder $\sH \in \Fq^{\sz\times \sy}$ defined in Definition~\ref{defi:lin_spir},
the answers of $\PSPIR$ are written as 
\begin{align}
\begin{pmatrix}
D_1\\
\vdots\\
D_{\snn}
\end{pmatrix}
&= \sQ M + T
\stackrel{\mathclap{(a)}}{=}
 \sQ' M_1 + \sH \RandSPIR
    = (\sQ' ,  \sH )
    \begin{pmatrix}
    M_1\\
    \RandSPIR
    \end{pmatrix} 
    \in \Fq^{\sz},
    \\
D_{j}
    &= (\sQ' ,  \sH)_{\tau^{-1}(j)}
    \begin{pmatrix}
    M_1\\
    \RandSPIR
    \end{pmatrix}  
    \in \Fq^{|\tau^{-1}(j)|},
    \label{share}
\end{align}
where 
$(a)$ follows from the condition $M_2 = \cdots =  M_\sff = 0$ of Corollary~\ref{coro:linearSPIRtoNSS}.


Next, we prove that the NSS protocol $\hPNSS[\PSPIR]$ is linear.
The shares of $\hPNSS[\PSPIR]$ are generated as $S_j = D_j$ in \eqref{share} while the secret $L$ of NSS is embedded as $M_1 = L \in \Fq^{\sx}$.
Thus, $\hPNSS[\PSPIR]$ corresponds to the linear NSS protocol with 
    the dealer's private randomness $\RandNSS =  \RandSPIR \in \Fq^{\sy}$,
    linear encoder  
    $(\sQ' ,  \sH )  \in \Fq^{\sz \times (\sx+\sy)}$,
    and 
    the same position map $\tau$ as $\PSPIR$.

\section{Proof of Conversion from linear CSNSS to MMSP} \label{sec:NSS=MMSP}




%
%

In this section, 
    we prove Theorem~\ref{theo:NSS=MMSP}, i.e.,
	the MMSP $\hPMMSP[\PNSS]$ defined in Protocol~\ref{prot:NSStoMMSP} from a linear $(\ACC,\REJ,\snn)$-CSNSS protocol is an $(\ACC,\REJ,\snn)$-MMSP.
Before the proof, 
	we prepare the following proposition and lemma.
	
\begin{prop} \label{prop:rvrank}
For any random variable $X \in \FF_q^{n}$ and
	$\sA \in\FF_q^{m\times n}$,
we have
\begin{align}
H(\sA X) \leq \rank \sA \log q. \label{eq:rankss}
\end{align}
When the distribution of $X$ is uniform, the equivalence of \eqref{eq:rankss} holds.
\end{prop}

\begin{lemm} \label{prop:rej_prop}
The rejection condition of an MMSP $\cP$ is equivalent to 
    \begin{align}
	\rank \sG_{\tau^{-1}(\cB)}''
	= \rank \sG_{\tau^{-1}(\cB)}
	\quad \forall\cB\in\REJ.
	\label{eq:rej_prop}
	\end{align}
\end{lemm}

\begin{proof}
First, from definition, the rejection condition of the MMSP $\cP$ is equivalent to
\begin{align}
\sx  + \rank\sG_{\tau^{-1}(\cB)}
=
\rank
\begin{pmatrix}
\sI_{\sx} & \sO_{\sx\times \sy} \\
\multicolumn{2}{c}{\sG_{\tau^{-1}(\cB)} } 
\end{pmatrix}
	\quad \forall\cB\in\REJ.
	\label{eq:equiv_sdresj}
\end{align}
Thus, we prove the equivalence between 
   \eqref{eq:rej_prop} and \eqref{eq:equiv_sdresj}.

The direction from \eqref{eq:rej_prop} to \eqref{eq:equiv_sdresj} is proved as
\begin{align}
&\sx  + \rank\sG_{\tau^{-1}(\cB)}
= \sx  + \rank\sG_{\tau^{-1}(\cB)}''\\
&=
\rank
\begin{pmatrix}
\sI_{\sx} & \sO_{\sx\times \sy} \\
\sO_{\sz\times \sx} & \sG_{\tau^{-1}(\cB)}''
\end{pmatrix}
=
\rank
\begin{pmatrix}
\sI_{\sx} & \sO_{\sx\times \sy} \\
\multicolumn{2}{c}{\sG_{\tau^{-1}(\cB)} } 
\end{pmatrix}
\end{align}
which implies \eqref{eq:equiv_sdresj}.
The direction from \eqref{eq:equiv_sdresj} to \eqref{eq:rej_prop} is proved as
\begin{align}
&\sx  + \rank\sG_{\tau^{-1}(\cB)}
\stackrel{\mathclap{(a)}}{=}
\rank
\begin{pmatrix}
\sI_{\sx} & \sO_{\sx\times \sy} \\
\multicolumn{2}{c}{\sG_{\tau^{-1}(\cB)} } 
\end{pmatrix}
\\
&=
\rank
\begin{pmatrix}
\sI_{\sx} & \sO_{\sx\times \sy} \\
\sO_{\sz\times \sx} & \sG_{\tau^{-1}(\cB)}''
\end{pmatrix}
= \sx  + \rank\sG_{\tau^{-1}(\cB)}'',
\end{align}
where $(a)$ follows from the rejection condition.
Thus, we obtain the desired statement.
\end{proof}

Now, we prove Theorem~\ref{theo:NSS=MMSP}.
\begin{proof}[Proof of Theorem~\ref{theo:NSS=MMSP}]
The rates of $\hPMMSP[\PNSS]$ and $\PNSS$ are trivially $\sx/\sz$.
In the following, we separately prove that $\hPMMSP[\PNSS]$ accepts $\fA$ and rejects $\fB$.

First, we prove that the MMSP $\hPMMSP[\PNSS]$ accepts $\ACC$.
Let $\cA \in \ACC$.
Then, the completely secure correctness condition of the linear CSNSS guarantees that
    there exists a function $h$ such that $h(S_{\cA}) = L$, i,e., 
	\begin{align}
	h(S_{\cA}) = h\paren*{\sG_{\tau^{-1}(\cA)} 
		\begin{pmatrix}
		L \\
		R
		\end{pmatrix}
		} 
	 = L.
	 \label{eq:liaaa}
	\end{align}
	Since \eqref{eq:liaaa} holds for any $L\in\FF_q^{\sx}$ and $R\in\FF_q^{\sy}$,
		$h$ is written as a linear map from $\Im \sG_{\tau^{-1}(\cA)}$ to $\FF_q^{\sx}$ and
			the associated matrix $\sK_h\in \FF_q^{\sx \times |\tau^{-1}(\cA)|}$ of the function $h$ satisfies
	\begin{align}
	\sK_h \sG_{\tau^{-1}(\cA)} =  
		(\sI_{\sx},  \sO_{\sx\times \sy}),
	\end{align}
	where $\sI_{\sx}$ is the $\sx\times\sx$ identity matrix and $\sO_{\sx\times \sy}$ is the $\sx\times \sy$ zero matrix.
	Thus, the row space of $\sG_{\tau^{-1}(\cA)}$ includes the row space of $(\sI_{\sx}, \sO_{\sx\times \sy})$, 
		which implies that the MMSP $\hPMMSP[\PNSS]$ accepts $\ACC$.  
%
%

Next, we prove that the MMSP $\hPMMSP[\PNSS]$ rejects $\fB$.
We denote $\sG_{f}  = (\sG_{f}', \sG_{f}'') \in \FF_q^{\sz\times \sx} \times \FF_q^{\sz\times \sy}$.
From Lemma~\ref{prop:rej_prop}, the fact that $\hPMMSP[\PNSS]$ rejects $\REJ$ is equivalent to 
    \begin{align}
	\rank \sG_{f,\tau^{-1}(\cB)}''
	= \rank \sG_{f,\tau^{-1}(\cB)}
	\quad \forall\cB\in\REJ.
	\label{eq:MMSPEQVDL}
	\end{align}
Thus, we prove \eqref{eq:MMSPEQVDL} as follows:
For any $\cB\in\REJ$,
	the MMSP $\hPMMSP[\PNSS]$ satisfies
	\begin{align}
	&\rank \sG_{f,\tau^{-1}(\cB)} \log q
	\stackrel{\mathclap{(a)}}{=} H(S_\cB)
	\stackrel{\mathclap{(b)}}{=}
	H(S_\cB | L = \bm{\ell})\\
	&= H(\sG_{f,\tau^{-1}(\cB)}' L + \sG_{\tau^{-1}(\cB)}'' R | L = \bm{\ell})\\
	&= H(\sG_{f,\tau^{-1}(\cB)}'' R)
	\stackrel{\mathclap{(c)}}{=} \rank \sG_{f,\tau^{-1}(\cB)}'' \log q
	\label{eq:simHrank2}
	\end{align}
	where $(a)$ and $(c)$ follow from Proposition~\ref{prop:rvrank} and the uniform randomness of $L$ and $R$,
	and $(b)$ follows 
	from the secrecy condition $ I(L ; S_\cB  ) = 0$ of the CSNSS protocol.
Therefore, the MMSP $\hPMMSP[\PNSS]$ rejects $\fB$.
\end{proof}

\section{Proof of Conversion from MMSP to projected  linear CSSPIR}
	\label{sec:NSStoSPIR}

In this section, 
	the SPIR protocol $\hPSPIR[\cP]$ defined in Protocol~\ref{prot:MMSPtoSPIR} from an $(\REC,\COL,\snn)$-MMSP $\cP$
		is an $(\REC,\COL,\snn,\sff)$-CSSPIR protocol.

For the proof, we prepare the following lemma.
\begin{lemm} \label{lemm:secrecyy}
Let $\cP  = (\sG, \tau)$ be an $(\ACC,\REJ,\snn)$-MMSP,
    $X\in\FF_q^\sx$ be a random vector,
    and $Y\in\FF_q^\sy$ be a uniform random vector.
For $\cX\in[\snn]$, we define 
\begin{align}
Z_{\tau^{-1}(\cX)}  =
			\sG_{\tau^{-1}(\cX)}  
\begin{pmatrix}
		X \\
		Y
\end{pmatrix}.
\end{align}
Then, the relation 
	\begin{align}	
	H(Z_{\tau^{-1}(\cB)}) = H(Z_{\tau^{-1}(\cB)} | X= \mathbf{x})
	\end{align} 
	holds for any $\cB\in\REJ$ and $\mathbf{x} \in \FF_q^\sx$, i.e., $I(Z_{\tau^{-1}(\cB)};X) = 0$.
\end{lemm}
\begin{proof}
From Lemma~\ref{prop:rej_prop}, the fact that $\cP$ rejects $\REJ$
	is equivalent to 
	\begin{align}
	\rank \sG_{\tau^{-1}(\cB)}''
	= \rank \sG_{\tau^{-1}(\cB)}
	\quad \forall\cB\in\REJ.
	\label{eq:ranksame}
	\end{align}
From this relation, we obtain 
	\begin{align}
	&H(Z_{\tau^{-1}(\cB)} | X = \mathbf{x})
	= H(\sG_{\tau^{-1}(\cB)}' \mathbf{x} + \sG_{\tau^{-1}(\cB)}'' Y )\\
	&= H(\sG_{\tau^{-1}(\cB)}'' Y)
	\stackrel{\mathclap{(a)}}{=} \rank \sG_{\tau^{-1}(\cB)}'' \log q\\
	&\stackrel{\mathclap{(b)}}{=} \rank \sG_{\tau^{-1}(\cB)} \log q
	\stackrel{\mathclap{(c)}}{=} H(Z_{\tau^{-1}(\cB)} ), 
	\label{eq:simHrank}
	\end{align}
	where $(a)$ and $(c)$ follow from Proposition~\ref{prop:rvrank}
	and $(b)$ follows from \eqref{eq:ranksame}.
%
\end{proof}

Now, we prove Theorem~\ref{theo:MMSPtoSPIR}.
\begin{proof}[Proof of Theorem~\ref{theo:MMSPtoSPIR}]
From the definition of Protocol~\ref{prot:MMSPtoSPIR},
    it is clear that the SPIR protocol $\hPSPIR[\cP]$ is projected linear and 
    the PIR rate is $\RPIR = \sx/\sz = \RMMSP$.
Thus, in the following, we separately prove that the SPIR protocol $\hPSPIR[\cP]$ has completely secure  
  user secrecy, server secrecy, and correctness.

First, the user secrecy of $\hPSPIR[\cP]$ is proved as follows.
The queries indexed by $\cB\in\COL$ are written as 
			\begin{align}
			Q_{\tau^{-1}(\cB)} 
			=
			\sG_{\tau^{-1}(\cB)}  
		\begin{pmatrix}
		\sE_k \\
		R
		\end{pmatrix}.
			\end{align}
	From Lemma~\ref{lemm:secrecyy} and the uniform randomness of $R$, 
		the distribution of $Q_{\tau^{-1}(\cB)}$ does not depend on the value of $k$, which implies the completely secure user secrecy.

Server secrecy of $\hPSPIR[\cP]$ is proved as follows. 
The answers indexed by $\cX\subset [\snn]$ are written as 
	\begin{align}
	D_{\cX}
	&= Q_{\tau^{-1}(\cX)} M + \sG_{\tau^{-1}(\cX)}'' \RandSPIR \\
	&= \sG_{\tau^{-1}(\cX)}' \sE_k M + \sG_{\tau^{-1}(\cX)}'' (R M + \RandSPIR)\\
	&= \sG_{\tau^{-1}(\cX)}'  M_k + \sG_{\tau^{-1}(\cX)}'' \RandSPIR'
	\\
	&= \sG_{\tau^{-1}(\cX)}
			\begin{pmatrix}
			M_k\\
			\RandSPIR'
			\end{pmatrix}
	\in \FF_q^{|\cX|},
	\label{eq:serv_noinfo}
	\end{align}
	where $\RandSPIR'  \coloneqq R M + \RandSPIR \in \FF_q^{\sy}$.
Since $\RandSPIR$ is uniform randomness, $\RandSPIR'$ in \eqref{eq:serv_noinfo} is also uniformly random.
Thus, for any $\cX\in[\snn]$,
		the user's received information \eqref{eq:serv_noinfo} 
		does not depend on the files except for $M_k$, which implies the completely secure server secrecy.

Next, we prove the correctness of the SPIR protocol $\hPSPIR[\cP]$.
Since $\cP$ accepts $\ACC$, i.e., 
	the row space of $\sG_{\tau^{-1}(\cA)}$ contains $\mathbf{e}_1,\ldots, \mathbf{e}_{\sx}\in\Fq^{\sx+\sy}$ for any $\cA\in\ACC$,
	there exists a matrix $\sK[\cA]$ such that 
	\begin{align}
	\sK[\cA] \sG_{\tau^{-1}(\cA)} =  
		(\sI_{\sx},  \sO_{\sx\times \sy}),
	\end{align}
	where $\sI_{\sx}$ is the $\sx\times\sx$ identity matrix and $\sO_{\sx\times \sy}$ is the $\sx\times \sy$ zero matrix.
Since \eqref{eq:serv_noinfo} implies 
\begin{align}
	D_{\cA} &= \sG_{\tau^{-1}(\cA)}
			\begin{pmatrix}
			M_k\\
			\RandSPIR'
			\end{pmatrix},
\end{align}
	by applying $\sK[\cA]$ to the answers $\cD_{\cA}$,  
	the user obtains $\sK[\cA]\cD_{\cA} = M_k$ correctly.
\end{proof}

\section{Examples of Constructions}\label{sec:examples}

\subsection{A construction with access structure in Example~\ref{exam:as}}
    \label{sub:examples}

In this subsection, 
    we give 
        a simple example of a MMSP 
    with the general access structure in Example~\ref{exam:as}.
    
The access structure in Example~\ref{exam:as} is defined as 
    $\snn = 3$, $\fA = \{  \{2,3\}, \{1,2,3\} \}$, and $\fB = \{\emptyset,  \{1\},  \{2\}, \{3\}, \{1,2\} \}$.
We fix $\sx = 1$, $\sy =2$, and $\sz = 4$, and define an $(\ACC,\REJ,\snn)$-MMSP by 
\begin{align}
\sG &= 
\begin{pmatrix}
0 & 1 & 2 \\
1 & 1 & 1 \\
0 & 1 & 1 \\
1 & 1 & 0 
\end{pmatrix}
\in\FF_3^{4\times 3}
,
\\
\tau &= \{ 1\mapsto1,\ 2\mapsto2,\ 3\mapsto3,\ 4\mapsto3 \}.
\end{align}
Then, we can confirm that this MMSP $\cP = (\sG,\tau)$ accepts $\fA$ and rejects $\fB$ as follows. 
To confirm that $\cP$ accepts $\fA$, 
    it is enough to confirm if
    \begin{align}
    \cE = \{(1,0,0)\} \subset 
    \rowspan \sG_{\cA}
    \end{align}
    for $\cA = \{2,3\}$.
   Since  
     \begin{align}
    \sG_{\tau^{-1}( \{2,3\} )} 
    =
    \sG_{\{2,3,4\}} 
    = 
    \begin{pmatrix}
    1 & 1 & 1 \\
    0 & 1 & 1 \\
    1 & 1 & 0 
    \end{pmatrix}
    \end{align}
satisfies the above inclusion, $\cP$ accepts $\fA$.
To confirm that $\cP$ accepts $\fB$, 
    we need to confirm 
        $$\spann \cE \cap \rowspan \sG_{\tau^{-1}(\cB)} = \{0 \}$$ for $\cB = \{1,2\}$ and $\cB=\{3\}$.
Since both 
    \begin{align}
    \sG_{\tau^{-1}(\{1,2\})} &= 
    \sG_{\{1,2\}} = 
    \begin{pmatrix}
    0 & 1 & 2 \\
    1 & 1 & 1 \\
    \end{pmatrix},
    \\
    \sG_{\tau^{-1}(\{3\})} &= 
    \sG_{\{3,4\}}= 
    \begin{pmatrix}
    0 & 1 & 1 \\
    1 & 1 & 0 
    \end{pmatrix}
    \end{align}
satisfy the above property, $\cP$ rejects $\fB$.
Therefore, $\cP$ is an $(\ACC,\REJ,\snn)$-MMSP.
The rate of $\cP$ is $\RMMSP = \sx/\sz = 1/4$.

From this MMSP $\cP$ and 
    Protocols~\ref{prot:SPIRtoSS}, \ref{prot:NSStoMMSP}, and \ref{prot:MMSPtoSPIR},
    we can also obtain $(\ACC,\REJ,\snn,\sff)$-SPIR and $(\ACC,\REJ,\snn)$-NSS protocols for the access structure in Example~\ref{exam:as} and any $\sff\geq2$.
The induced SPIR and NSS protocols has the same rate $\sx/\sz = 1/4$ and randomness rate $\sy/\sx = 2$.

\subsection{Optimal construction of $(\srr,\stt,\snn)$-CSNSS and $(\srr,\stt,\snn,\sff)$-CSSPIR} \label{sec:optimal}

In this subsection, we construct $(\srr,\stt,\snn)$-CSNSS and $(\srr,\stt,\snn,\sff)$-CSSPIR protocols with optimal rates
	from a $(\srr,\stt,\snn)$-MMSP. 
The proposed protocols are already proposed by Yamamoto \cite{Yamamoto86}
	and by Tajeddine et al. \cite{TGKFH19} (as a special case), but we construct these protocols from MMSPs.


In the following, we assume that the size of the finite field $\FF_q$ is at least $\snn$.
We define an MMSP $\cP = (\sG, \tau)$
	by the Vandermonde matrix, which is also the generator matrix of a Reed-Solomon code \cite{RS60}.
Let $(\sx,\sy,\sz) = ( \srr-\stt , \stt , \snn) $ and $\tau = \id_{[\snn]}$.
Define 
\begin{align*}
\sG = 
	\begin{pmatrix}
	1 & \alpha_1 & \cdots & \alpha_{\sx+\sy-1}	\\
	1 & \alpha_1^2 & \cdots & \alpha_{\sx+\sy-1}^2	\\
	\vdots & \vdots	 & \ddots & \vdots	\\
	1 & \alpha_1^{\sz} & \cdots & \alpha_{\sx+\sy-1}^{\sz}
	\end{pmatrix} 
	=
	\begin{pmatrix}
	1 & \alpha_1 & \cdots & \alpha_{\srr-1}	\\
	1 & \alpha_1^2 & \cdots & \alpha_{\srr-1}^2	\\
	\vdots & \vdots	 & \ddots & \vdots	\\
	1 & \alpha_1^{\snn} & \cdots & \alpha_{\srr-1}^{\snn}
	\end{pmatrix} 
	.
\end{align*}
Then, we can easily confirm that the MMSP $\cP$ accepts $\ACC = \{ \cA \subset [\snn] \mid |\cA| \geq \srr \}$
	and rejects 
	$\REJ = \{ \cB \subset [\snn] \mid |\cB| \leq \stt \}$ as follows.
For any permutation $\pi$ of $[\snn]$, 
    fix $\cA \coloneqq \{\pi(i) | i \in[\srr] \}\in\fA $
    and
    $\cB \coloneqq \{\pi(i) | i \in[\stt] \}\in\fB $.
The matrices
	\begin{align*}
 	\sS &=
	\sG_{\tau^{-1}(\cA)} = 
	\sG_{\cA} = 
	\begin{pmatrix}
	1 & \alpha_1^{\pi(1)} & \cdots & \alpha_{\srr-1}^{\pi(1)}	\\
	\vdots & \vdots	 & \ddots & \vdots	\\
	1 & \alpha_1^{\pi(\srr)} & \cdots & \alpha_{\srr-1}^{\pi(\srr)}
	\end{pmatrix} 
	\in\FF_q^{\srr\times\srr}
	,
	\\
	\sT &=
	\begin{pmatrix}
	\sI_{\srr-\stt}  & 0	\\
	\multicolumn{2}{c}{\sG_{\tau^{-1}(\cB)} }
	\end{pmatrix}
	=
	\begin{pmatrix}
	\sI_{\srr-\stt}  & 0	\\
	\multicolumn{2}{c}{\sG_{\cB} }
	\end{pmatrix}
	=
	\begin{pmatrix}
	\multicolumn{3}{c}{\sI_{\srr-\stt}}  & 0	\\
	\hline
	1 & \alpha_1^{\pi(1)} & \cdots & \alpha_{\srr-1}^{\pi(1)}	\\
	\vdots & \vdots	 & \ddots & \vdots	\\
	1 & \alpha_1^{\pi(\stt)} & \cdots & \alpha_{\srr-1}^{\pi(\stt)}
	\end{pmatrix} 
	\in\FF_q^{\srr\times\srr}
	\end{align*}
	are invertible matrices.
We obtain the acceptance of $\cP$ because 
	the invertibility of $\sS$ implies that the row span of $\sS$ is $\FF_q^{\srr}$ and thus includes $\cE = \{ \mathbf{e}_1,\ldots, \mathbf{e}_{\srr-\stt} \}$.
We obtain the rejection of $\cP$ because 
    the row vectors of $\sT$ are linearly independent and 
    thus, 
    the span of the last $\stt$ row vectors does not include the span of $\cE$.

Thus, 
	the CSSPIR protocol $\hPSPIR[\cP]$ defined in Protocol~\ref{prot:MMSPtoSPIR} is 
		an $(\srr,\stt,\snn)$-SPIR protocol with the PIR rate $\RPIR = (\srr-\stt)/\snn$.
	A linear $(\srr,\stt,\snn)$-CSNSS protocol with the SS rate $\RSS = \srr-\stt$ can also be constructed by the equivalence in Corollary~\ref{coro:equiv}.
Given $\stt,\srr,\snn$,
	these protocols are optimal from Proposition~\ref{prop:rssup} and Corollary~\ref{coro:cap}. 

\section{Conclusion} \label{sec:conclusion}

We have studied the equivalence relation between non-perfect secret sharing and symmetric private information retrieval.
We have defined the two protocols with access structures, 
	which represent the authorized and forbidden sets for non-perfect secret sharing
		and the response and collusion patterns for symmetric private information retrieval.
We first showed that any SPIR protocols can be converted into NSS protocols.
From this relation, we proved an upper bound of CSSPIR capacity with arbitrary response and collusion patterns.
Next, 
	we proved the equivalence of linear CSNSS, linear CSSPIR, and MMSP.
From this implication, we obtained the existence of CSSPIR for any access structure.

%
%


%

\appendices

\section{Proof of Theorem~\ref{theo:MDS-MMSP}} \label{MMSP-MDS}

We separately prove $(a) \implies (b)$ and $(b) \implies (a)$ of Theorem~\ref{theo:MDS-MMSP}.

\noindent\textit{Step 1 $(a) \implies (b)$}:\quad 
It follows from the rejection condition of the MMSP that
        any $\stt$ rows of $\sG''$ are linearly independent. 
Thus, $\sG''$ is the generator matrix of an $(\snn,\stt)$-MDS code.
Also, we can show that any $\srr$ rows of $\sG$ are linearly independent
	from the acceptance condition of the MMSP and the MDS property of $\sG''$.
To be precise, without losing generality,
	we prove that $\sG_{[\srr]}\in\FF_q^{\srr\times \srr}$ is invertible as follows.

The acceptance condition implies $\spann \cE \subset \rowspan \sG_{[\srr]}$.
Thus, $\spann \cE \oplus \rowspan \sG_{[\stt]} \subset  \rowspan \sG_{[\srr]}$.
On the other hand, the rejection condition implies $\rowspan \sG_{[\stt]} \cap \rowspan \cE = \{0\}$.
Thus, we have $\srr = \dim (\spann \cE \oplus \rowspan \sG_{[\stt]}) 
	\leq  \dim \sG_{[\srr]} \leq 
	\srr$,
	which implies $\sG_{[\srr]}$ is invertible.

\noindent\textit{Step 2 $(b) \implies (a)$}:\quad 
We prove in the following that  $\sG$ and $\sG''$ with condition (b) satisfy the acceptance and rejection conditions of the MMSP $\cP$.
Since any $\srr$ rows of $\sG\in\FF_q^{\snn\times \srr}$ are linearly independent, 
    the space spanned by those rows is $\FF_q^{\srr}$, which implies the acceptance condition of the MMSP $\cP$.
    
Next, we prove the rejection condition of $\cP$ from the MDS property of $\sG''\in\FF_q^{\snn\times \stt}$ by contradiction.
Suppose that the rejection condition does not hold, i.e., 
    there exists 
        a row vector $\bx \in \Fq^{\snn}$ 
            and a set $\cB$ with $|\cB| =\stt$
            satisfying $0 \neq \bx \in \spann \cE \cap \rowspan \sG_\cB$. 
Since $\bx \in \spann \cE$, the last $\stt$ coordinates of $\bx$ are $0$.
On the other hand, since $\sG''_{\cB}\in\Fq^{\stt\times \stt}$ is invertible,  
    the last $\stt$ coordinates of $\bx$ are $0$ if and only if $\bx = 0$, which is a contradiction.
Thus, 
the rejection condition holds.

\section{Proof of Lemma~\ref{lemm:noleak}} \label{append:noleak}

In this section, we prove of Lemma~\ref{lemm:noleak}.
In the following subsections, we separately prove the following inequalities. 
\begin{align}
	&I(M_{[2:\sff]} ; M_1 D |R=r^{\ast}, K=1) \nonumber \\
	&\leq  
     \beta + h_2(1-\alpha) + (1-\alpha)\log \smm + \log \alpha, \label{eq:nol1} \\
	&I(M_1 ; D_{\cB} |R=r^{\ast}, K=1) \nonumber  \\
	&\leq \beta +  4\sqrt{2\gamma\sff} \log \smm + 2h_2(\sqrt{2\gamma\sff})
	\label{eq:nol2}.   
\end{align}
Then, with these two inequalities, we obtain the lemma as 
\begin{align}
& I(M_1 ; D_{\cB} | M_{[2:\sff]}, R=r^{\ast}, K=1) \\
&\le I(M ; D_{\cB} |R=r^{\ast}, K=1) \\
&= I(M_1 ; D_{\cB} |R=r^{\ast}, K=1) + I(M_{[2:\sff]} ; D |M_1 , R=r^{\ast}, K=1)\\
&= I(M_1 ; D_{\cB} |R=r^{\ast}, K=1) + I(M_{[2:\sff]} ; M_1 D |R=r^{\ast}, K=1)\\
&\leq \xi(\alpha,\beta,\gamma),
\end{align}
where the last inequality follows from \eqref{eq:nol1} and \eqref{eq:nol2}.
%
%

\subsection{Proof of \eqref{eq:nol1}}

For the proof of \eqref{eq:nol1}, we prove the following lemma.

\begin{lemm} \label{lemm:sdfsqwqqasc}
Suppose 
\begin{itemize}
\item $A,B$ are independent,
\item $B$ is recovered from $C$ with probability $\alpha$, and
\item $I(A;C) \le \beta$.
\end{itemize}
Then, we have
\begin{align}
I(A;BC) \leq \beta +  h_2(1-\alpha) + (1-\alpha)\log |\cB|  + \log \alpha,
\end{align}
where $\cB$ is the space of $B$.
\end{lemm}

By applying the above lemma for the case of 
    $(A,B,C) = (M_{[2:\sff]}, M_1, D)$ while conditioning $R=r^{\ast}, K=1$,
 Eq.~\eqref{eq:nol1} is obtained as 
    \begin{align}
    &I(M_{[2:\sff]} ; M_1 D_{\cB} |R=r^{\ast}, K=1) \\
    &\leq I(M_{[2:\sff]} ; M_1 D |R=r^{\ast}, K=1) \\
    &\leq \beta + h_2(1-\alpha) + (1-\alpha)\log \smm + \log \alpha.
    \end{align}

\begin{proof}[Proof of Lemma~\ref{lemm:sdfsqwqqasc}]
From Fano's inequality, we have
\begin{align}
H(B|C)  \leq  \beta +  h_2(1-\alpha) + (1-\alpha)\log |\cB| . 
\end{align}
On the other hand, from the lower bound of the guessing probability \cite{BCHLW06,KRS09}, we have 
\begin{align}
H(B|AC) \ge -\log \alpha .
\end{align}
Thus, combining the above two inequalities, we obtain the desired inequality as
\begin{align}
I(A;BC) &= 
    I(A;C) + I(A;B|C)\\
    &= I(A;C) + H(B|C) - H(B|AC)\\
    &\leq 2 \beta +  h_2(1-\alpha) + (1-\alpha)\log |\cB| + \log \alpha.
\end{align}

\end{proof}

\subsection{Proof of \eqref{eq:nol2}}

For the proof, we define the variational distance
\begin{align}
d(p,q) \coloneqq \sum_{x}  |p(x)-q(x)|.
\end{align}
Throughout this section,
	we denote the distribution of a random variable $X$ by $P_X$ and the probability $P_X(x)$ by $P_x$, i.e., subscript with the lowercase letter of $X$.

First, we prepare the following proposition.
%
%
\begin{prop}[Classical Alicki-Fannes inequality {\cite{AF04}}] \label{prop:AF_classical}
Let $X,Y$ be random variable on $\cX, \cY$
	 and $p,q$ be probability distributions on $\cX\times \cY$ such that $\epsilon \coloneqq d(p,q) = \sum_{x,y}  |p(x,y)-q(x,y)| $.
Then,	
	\begin{align}
	|H(p_{X|Y}) - H(q_{X|Y}) | \leq 4\epsilon \log |\cX| + 2h_2(\epsilon).
	\end{align}
\end{prop}
\begin{remark}
In \cite{AF04}, Proposition~\ref{prop:AF_classical} is originally proved for quantum systems and states.
By restricting the quantum systems and states as classical systems and random variables, 
	we directly obtain Proposition~\ref{prop:AF_classical}.
\end{remark}

Now, we prove \eqref{eq:nol2}.
\begin{proof}[Proof of  \eqref{eq:nol2}]
We prove the lemma by two steps.

\noindent\textbf{Step 1}:\quad 
First, we prove 
\begin{align}
 d(P_{M_1D_{\cB}Q_{\cB} | K=1}, P_{M_1D_{\cB}Q_{\cB} | K=2}) \leq \sqrt{2\gamma\sff}
 \label{eq:dsdtlappe}
\end{align}
for any $\cB\in\COL$.
This inequality for the threshold case is proved in \cite[Lemma 5]{SH21-collude}.
With the similar idea, we give the proof of \eqref{eq:dsdtlappe}.
We have
\begin{align}
\gamma 
	&\geq I(K;Q_{\cB})\\
	&= D(P_{KQ_{\cB}}  \| P_{K} \times P_{Q_{\cB}})\\
	&\stackrel{\mathclap{(a)}}{=} \frac{1}{\sff} \sum_{k} D(P_{Q_{\cB}|K=k}  \| P_{Q_{\cB}})\\
	&\stackrel{\mathclap{(b)}}{\geq} \frac{2}{\sff} \sum_{k} d^2(P_{Q_{\cB}|K=k}  , P_{Q_{\cB}})\\
	&\geq  \frac{2}{\sff} d^2(P_{Q_{\cB}|K=k'}  , P_{Q_{\cB}}) \\
	&\stackrel{\mathclap{(c)}}{=}  \frac{2}{\sff}  d^2(P_{M_1D_{\cB}Q_{\cB}|K=k'}  , P_{M_1D_{\cB}Q_{\cB}})
		\label{eq:qeqfds}
\end{align}
for any $k' \in[\sff]$.
The equality $(a)$ follows from the uniform randomness of $K$,
    and 
    the inequality $(b)$ follows from Pinsker's inequality.
The equality $(c)$ follows from
	\begin{align}
	&d(P_{M_1D_{\cB}Q_{\cB}|K=k}  , P_{M_1D_{\cB}Q_{\cB}}) \\
	&= \frac{1}{2} \sum_{m_1,a_{\cB}, q_{\cB}, k}  | P_{m_1a_{\cB}q_{\cB}|k}  - P_{m_1a_{\cB}q_{\cB}} | \\
	&= \frac{1}{2} \sum_{m_1,a_{\cB}, q_{\cB}, k}  | P_{m_1a_{\cB}|kq_{\cB}} P_{q_{\cB}|k}  - P_{m_1a_{\cB}|q_{\cB}} P_{q_{\cB}}| \\
	&\stackrel{\mathclap{(d)}}{=} \frac{1}{2} \sum_{m_1,a_{\cB}, q_{\cB}, k}  | P_{m_1a_{\cB}|q_{\cB}} P_{q_{\cB}|k}  - P_{m_1a_{\cB}|q_{\cB}} P_{q_{\cB}}| \\
	&= \frac{1}{2} \sum_{q_{\cB}, k}  |  P_{q_{\cB}|k}  - P_{q_{\cB}}| \\
	&= d(P_{Q_{\cB}|K=k}  , P_{Q_{\cB}})
	\end{align}
	where the inequality $(d)$ holds since $K - Q_{\cB}- M_1D_{\cB}$ is a Markov chain.
Then, from Eq.~\eqref{eq:qeqfds} and the triangular inequality, we obtain \eqref{eq:dsdtlappe} as
	\begin{align}
	\sqrt{2\gamma\sff} &\geq 
	d(P_{M_1D_{\cB}Q_{\cB}|K=1} , P_{M_1D_{\cB}Q_{\cB}}) \\
	&\quad + 
	d(P_{M_1D_{\cB}Q_{\cB}|K=2} , P_{M_1D_{\cB}Q_{\cB}}) \\
	&\geq 
	d(P_{M_1D_{\cB}Q_{\cB}|K=1} , P_{M_1D_{\cB}Q_{\cB}|K=2} ).
	\end{align}


%

\noindent\textbf{Step 2}:\quad 
Next, we prove the desired inequality  
\begin{align}
I(M_1 ; D_{\cB} |R=r^{\ast}, K=1)  \leq \beta + g(\gamma, \smm),
\end{align}
where
\begin{align}
g(\gamma, \smm) &\coloneqq 4\sqrt{2\gamma\sff} \smm + 2h_2(\sqrt{2\gamma\sff}).
\end{align}
We have
\begin{align}
	&|I(M_1 ; D_{\cB} |R, K=1) - I(M_1 ; D_{\cB} |R, K=2) |\\
	&\stackrel{\mathclap{(a)}}{=} |I(M_1 ; D_{\cB} |Q_{\cB}, K=1) - I(M_1 ; D_{\cB} |Q_{\cB} ,K=2) |\\
	&= |H(M_1 |D_{\cB} Q_{\cB} , K=1) - H(M_1 | D_{\cB} Q_{\cB}, K=2) |	\\
	&\stackrel{\mathclap{(b)}}{\leq}  4\sqrt{2\gamma\sff} \log |\cM| + 2h_2(\sqrt{2\gamma\sff})\\
	& = g(\gamma,\smm),
		\label{eq:regaar}
\end{align}
where $(a)$ holds because $R-Q_{\cB}- M_iD_{\cB}$ is a Markov chain and 
    $(b)$ is obtained by combining Proposition~\ref{prop:AF_classical} and \eqref{eq:dsdtlappe}.
Rearranging the inequality \eqref{eq:regaar}, we have  
\begin{align}
		I(M_1 ; D_{\cB} |R, K=1) 
	&\leq I(M_1 ; D_{\cB} |R, K=2) + g(\gamma, \smm )	\nonumber \\
	&\leq I(M_{[\sff]\setminus 2} ; D |R, K=2) + g(\gamma, \smm ) \nonumber	\\
	&\stackrel{\mathclap{(c)}}{\leq} \beta + g(\gamma, \smm),
        \label{eq:lasteqqq}
\end{align}
where $(c)$ follows from the server secrecy of $\PSPIR$. Since $r^{\ast}$ is defined in \eqref{lemmeq:correct} to satisfy the inequality 
    \begin{align}
     & I(M_1 ; D_{\cB} |R= r^{\ast}, K=1) \\
     &\leq  \sum_{r\in\cR} P_R(r) I(M_1 ; D_{\cB} |R= r, K=1) \\
     &= I(M_1 ; D_{\cB} |R, K=1), 
    \end{align}
the inequality \eqref{eq:lasteqqq} derives the desired inequality 
    \begin{align}
    I(M_1 ; D_{\cB} |R= r^{\ast}, K=1) 
    \leq
    \beta + g(\gamma, \smm).
    \end{align}
\end{proof}


\end{document}

%% file: symbols_JSAC.tex
\usepackage[caption = false]{subfig}

\usepackage{tikz}
\usetikzlibrary{shapes,snakes}
\usetikzlibrary{tikzmark}
\usetikzlibrary{fit,backgrounds}
\usetikzlibrary{matrix,decorations.pathreplacing,angles,quotes}
\usetikzlibrary{patterns}
\usetikzlibrary{calc}
\usetikzlibrary{positioning}
\tikzstyle{block} = [rectangle, draw, 
    minimum width=4em, text centered, rounded corners, minimum height=1em]
\tikzstyle{rec} = [rectangle, draw]
\tikzstyle{line} = [draw, -latex]
\usetikzlibrary{intersections}

\DeclareFontFamily{OMX}{MnSymbolE}{}
\DeclareFontShape{OMX}{MnSymbolE}{m}{n}{
    <-6>  MnSymbolE5
   <6-7>  MnSymbolE6
   <7-8>  MnSymbolE7
   <8-9>  MnSymbolE8
   <9-10> MnSymbolE9
  <10-12> MnSymbolE10
  <12->   MnSymbolE12}{}
\DeclareSymbolFont{mnlargesymbols}{OMX}{MnSymbolE}{m}{n}
\SetSymbolFont{mnlargesymbols}{bold}{OMX}{MnSymbolE}{b}{n}
\DeclareMathDelimiter{\llangle}{\mathopen}{mnlargesymbols}{'164}{mnlargesymbols}{'164}
\DeclareMathDelimiter{\rrangle}{\mathclose}{mnlargesymbols}{'171}{mnlargesymbols}{'171}

\newtheorem{lemm}{Lemma}
\newtheorem{defi}{Definition}
\newtheorem{exam}{Example}
\newtheorem{theo}{Theorem}
\newtheorem{prot}{Protocol}
\newtheorem*{prot*}{Protocol}
\newtheorem{prop}{Proposition}

\newtheorem{coro}{Corollary}

\theoremstyle{definition}
\newtheorem{remark}{Remark}
\newtheorem*{remark*}{Remark}

\newcommand{\cA}{\mathcal{A}}
\newcommand{\cB}{\mathcal{B}}
\newcommand{\cC}{\mathcal{C}}
\newcommand{\cD}{\mathcal{D}}
\newcommand{\cE}{\mathcal{E}}

\newcommand{\cM}{\mathcal{M}}

\newcommand{\cP}{\mathcal{P}}
\newcommand{\cQ}{\mathcal{Q}}
\newcommand{\cR}{\mathcal{R}}
\newcommand{\cS}{\mathcal{S}}
\newcommand{\cT}{\mathcal{T}}
\newcommand{\cU}{\mathcal{U}}

\newcommand{\cX}{\mathcal{X}}
\newcommand{\cY}{\mathcal{Y}}

 \renewcommand{\Im}{\operatorname{Im}}

\DeclareMathOperator{\spann}{span}
\DeclareMathOperator{\rowspan}{rowspan}

\DeclareMathOperator{\id}{id}
\DeclareMathOperator{\pr}{Pr}

\DeclareMathOperator*{\argmax}{\arg\!\max}
\DeclareMathOperator*{\argmin}{\arg\!\min}
\DeclareMathOperator*{\rank}{\mathrm{rank}}
\DeclarePairedDelimiter\px{\{}{\}}
\DeclarePairedDelimiter\paren{(}{)}

\newcommand\sA{\mathsf{A}}
\newcommand\sB{\mathsf{B}}

\newcommand\sE{\mathsf{E}}

\newcommand\sG{\mathsf{G}}
\newcommand\sH{\mathsf{H}}
\newcommand\sI{\mathsf{I}}
\newcommand\sJ{\mathsf{J}}
\newcommand\sK{\mathsf{K}}

\newcommand\sO{\mathsf{O}}

\newcommand\sQ{\mathsf{Q}}

\newcommand\sS{\mathsf{S}}
\newcommand\sT{\mathsf{T}}

\newcommand\bx{\mathbf{x}}

\newcommand\saa{\mathsf{a}}
\newcommand\sbb{\mathsf{b}}

\newcommand\see{\mathsf{e}}
\newcommand\sff{\mathsf{f}}

\newcommand\skk{\mathsf{k}}
\newcommand\smm{\mathsf{m}}
\newcommand\snn{\mathsf{n}}

\newcommand\srr{\mathsf{r}}

\newcommand\stt{\mathsf{t}}

\newcommand\sx{\mathsf{x}}
\newcommand\sy{\mathsf{y}}
\newcommand\sz{\mathsf{z}}

\newcommand\PSPIR{\Phi^{\smm}_{\mathrm{SPIR}}}
\newcommand\PNSS{\Phi^{\smm}_{\mathrm{NSS}}}

\newcommand\hPSPIR{\hat{\Phi}^{\smm}_{\mathrm{SPIR}}}
\newcommand\hPNSS{\hat{\Phi}^{\smm}_{\mathrm{NSS}}}
\newcommand\hPMMSP{\hat{\cP}}

\newcommand\fA{\mathfrak{A}}
\newcommand\fB{\mathfrak{B}}

\newcommand\REC{\mathfrak{A}}
\newcommand\COL{\mathfrak{B}}
\newcommand\ACC{\mathfrak{A}}
\newcommand\REJ{\mathfrak{B}}

\newcommand\FF{\mathbb{F}}
\newcommand\Fq{\mathbb{F}_q}

\newcommand\RSS{R_{\mathrm{SS}}}
\newcommand\RPIR{R_{\mathrm{PIR}}}
\newcommand\RMMSP{R_{\mathrm{MMSP}}}

\newcommand\RandSPIR{U_{\mathrm{SPIR}}}
\newcommand\RandNSS{U_{\mathrm{NSS}}}

\newcommand\SRandSPIR{\cU_{\mathrm{SPIR}}}
\newcommand\SRandNSS{\cU_{\mathrm{NSS}}}

\newcommand\mast{m_{[2:\sff]}^\ast}